\documentclass{article}

\usepackage{arxiv}

\usepackage[utf8]{inputenc} 
\usepackage[T1]{fontenc}    
\usepackage{hyperref}       
\usepackage{url}            
\usepackage{booktabs}       
\usepackage{amsfonts}       
\usepackage{nicefrac}       
\usepackage{microtype}      
\usepackage{lipsum}
\usepackage{graphicx}

\usepackage{inputenc}
\usepackage{textcomp}

\usepackage{amsmath}
\usepackage[version=4]{mhchem}
\usepackage{longtable,tabularx}
\usepackage{amssymb}
\usepackage{cases}
\usepackage{color}

\usepackage{overpic}

\usepackage{array}
\usepackage{multirow}
\usepackage{bigdelim}
\usepackage{ctable}
\usepackage{mdwtab}

\usepackage{bm}

\newcommand{\bp}{\bm{\mathrm{p}}}

\DeclareMathOperator{\sign}{sgn}

\newcommand{\R}{\mathbb{R}}

\usepackage{accents}

\newtheorem{theorem}{Theorem}
\newtheorem{lemma}{Lemma}

\newtheorem{corollary}{Corollary}

\newtheorem{problem}{Problem}

\newtheorem{definition}{Definition}

\newenvironment{proof}{{\em Proof:}}{\hfill \hspace*{1pt} \hfill  $\blacksquare$}

\newenvironment{proofof}{{\em Proof of}}{\hfill \hspace*{1pt} \hfill  $\blacksquare$}

\newtheorem{remark}{Remark}

\title{Prying Pedestrian Surveillance-Evasion: Minimum-Time Evasion from an Agile Pursuer}

\author{
 Philipp Braun \\
  Australian National University\\
  Canberra, Australia \\
  \texttt{philipp.braun@anu.edu.au} \\
   \And
 Timothy L. Molloy \\
  Australian National University\\
  Canberra, Australia \\
  \texttt{timothy.molloy@anu.edu.au} \\
  \And
 Iman Shames \\
  Australian National University\\
  Canberra, Australia \\
  \texttt{iman.shames@anu.edu.au} \\
}

\begin{document}
\maketitle
\begin{abstract}
A new surveillance-evasion differential game is posed and solved in which an agile pursuer (the prying pedestrian) seeks to remain within a given surveillance range of a less agile evader that aims to escape.
In contrast to previous surveillance-evasion games, the pursuer is agile in the sense of being able to instantaneously change the direction of its velocity vector, whilst the evader is constrained to have a finite maximum turn rate.
Both the game of kind concerned with conditions under which the evader can escape, and the game of degree concerned with the evader seeking to minimize the escape time whilst the pursuer seeks to maximize it, are considered.
The game-of-degree solution is surprisingly complex compared to solutions to analogous pursuit-evasion games with an agile pursuer since it exhibits dependence on the ratio of the pursuer's speed to the evader's speed.
It is, however, surprisingly simple compared to solutions to classic surveillance-evasion games with a turn-limited pursuer.
\end{abstract}

\section{Introduction}

The widespread and accelerating adoption of agile vehicles, such as quadrotor drones, has motivated a plethora of new guidance problems. 
Particular interest and progress has centered on formulating and solving guidance problems as novel pursuit-evasion differential games with agile (or omnidirectional) pursuers capable of instantaneously altering their velocity vectors \cite{Exarchos2014,Exarchos2015,Exarchos2016,Molloy2020, Ruiz2013, Jha2019}.
For instance, a \emph{suicidal pedestrian} differential game was posed and solved in \cite{Exarchos2015,Exarchos2014,Exarchos2016} by reversing the maneuverability constraints of the pursuer and evader in the classic \emph{homicidal chauffeur} differential game \cite{Isaacs65,Merz1971} so that the pursuer is agile and the evader has a finite maximum turn rate. 
Similarly, \cite{Molloy2020} posed and solved a variation of a classic collision-avoidance game between two (noncooperative) ships \cite{Miloh1976,Olsder1978,Merz1973} in which one of the ships is agile.
These new games offer practical advantages by more accurately reflecting the (worst-case) maneuverability of agile vehicles. 
They also have remarkably simple solutions (i.e., Nash equilibrium trajectories) that can be fully characterized, unlike many classic differential games (cf.\ \cite{Merz1971,Isaacs65,Weintraub2020,Miloh1976,Olsder1978,Merz1973}). 
Surprisingly, however, agile pursuers have yet to be considered in \emph{surveillance-evasion} differential games, which are historic counterparts to pursuit-evasion games.

The first surveillance-evasion differential game in the open literature was introduced by Dobbie \cite{Dobbie1966} (see also \cite{Taylor1970}). 
He posed and solved a \emph{game of kind} examining conditions under which an agile evader can escape from a circular surveillance region centered on a turn-limited pursuer.
Lewin and Breakwell \cite{Lewin1975} subsequently posed and solved a \emph{game-of-degree} version in which the agile evader (resp.\ turn-limited pursuer) aims to minimize (resp. maximize) the time to escape from the circular surveillance region.
The game-of-degree solution was found to be ``surprisingly complex'' with a variety of singular arcs dividing the game space into regions where the pursuer and evader adopt different strategies (cf.\ \cite{Lewin1975}).
Variations of these original surveillance-evasion games have subsequently been examined, including with a conical surveillance region \cite{Lewin1979}; with both the evader and pursuer being turn-limited \cite{Greenfeld1987}; with the combination of a conical surveillance region and a turn-limited evader and pursuer \cite{Ruiz2024}; and, with the pursuer being a differential-drive robot whilst the evader is agile \cite{Ruiz2022,Saavedra2024}.
The closest existing work has come to considering an agile pursuer and a turn-limited evader appears to be the isotropic rocket surveillance-evasion game of \cite{Lewin1989}.
In \cite{Lewin1989}, whilst the evader is not turn-limited (it is agile in the sense of being able to instantaneously change its velocity), the pursuer is able to instantaneously change its acceleration (but not its velocity).
No works therefore appear to consider a surveillance-evasion game with an agile pursuer and a turn-limited evader.

Most recently, there has been growing interest in formulating and solving surveillance-evasion problems as optimal-control problems (or single-player differential games) \cite{Molloy2023,Weintraub2024,Weintraub2023,Weintraub2021,VonMoll2023}.
In these formulations, the strategy of only one of the evader or pursuer is optimized, while the other’s strategy remains fixed. 
For example, \cite{Molloy2023,Weintraub2024} consider the problem of controlling a turn-limited evader to enable it to escape from a stationary circular region in minimum time, which can be viewed as fixing the pursuer to be stationary and optimizing over evader strategies.
Similarly, \cite{VonMoll2023} considers the problem of controlling an evader to minimize its time under surveillance by a pursuer that employs a pure pursuit strategy.
Conversely, \cite{Weintraub2021,Weintraub2023} consider the problem of controlling a turn-limited pursuer to maximize the time that a constant-velocity (non-maneuvering) evader is under surveillance.
Interestingly, however, connections between many of these optimal control problems and analogous game-theoretic formulations are undeveloped.

The key contribution of this paper is the formulation and solution of a novel prying pedestrian surveillance-evasion differential game between an agile pursuer (the ``prying pedestrian'') that seeks to keep a turn-rate limited evader within a circular surveillance region, whilst the evader seeks to escape.
This game resembles the emergent problem of controlling a wheeled ground vehicle (e.g., a car) to escape surveillance by an agile aerial vehicle with a circular sensor footprint (e.g., a quadcopter drone at maximum altitude with a downward-facing camera).
It also extends the recent single-player or optimal-control surveillance problems of \cite{Molloy2023,Weintraub2024,VonMoll2023} to a differential game setting.
Finally, the solution of this game offers insight into the performance attainable by any type of pursuer since an agile pursuer can employ the surveillance strategies of any (less agile) pursuer.

We solve both the game of kind and game of degree forms of the prying pedestrian surveillance-evasion differential game.
The game of kind solution establishes conditions under which the evader can escape from the pursuer's detection region (or the pursuer can surveil the evader indefinitely), whilst the game of degree solution considers the time the evader is within the surveillance region as the payoff and establishes optimal (i.e., Nash equilibrium) strategies for both the evader and pursuer to minimize and maximize the payoff, respectively.
The game of kind solution is intuitively simple, with escape always being possible when the evader is faster than the pursuer (provided that the evader simply continues in a straight line), but impossible otherwise (provided that the pursuer employs a simple strategy to remain within range of the evader).
The game-of-degree solution is surprisingly complex, with the ratio of the pursuer's speed to the evader's speed determining the nature of several singular arcs that divide the game space into regions in which the evader either turns away from the pursuer or moves in a straight line, whilst the pursuer moves along piecewise-linear trajectories.
Interestingly, the solutions to comparable pursuit-evasion and collision-avoidance games with agile pursuers do not exhibit this speed-ratio dependence (cf.\ \cite{Exarchos2015} and \cite{Molloy2020}).
Nevertheless, the solution of the prying pedestrian surveillance-evasion game is simpler than that of the original surveillance-evasion game with a turn-limited pursuer (cf.\ \cite{Dobbie1966,Lewin1975}).

The paper is organized as follows.
In Section \ref{sec:problem}, we formulate the prying pedestrian surveillance-evasion games of kind and degree, and solve the game of kind.
In Section \ref{sec:solution_game_of_degree}, we solve the game of degree.
In Section \ref{sec:analysis_of_solution}, we illustrate and discuss properties of the solution to the game of degree, including its specialization to (single-player) minimum-time circle escape solutions and its degree of optimality over alternative approaches such as pure pursuit.
Finally, we present conclusions in Section \ref{sec:conclusion}.

\emph{Notation:} For $\xi\in \R^n$, $|\xi|=\sqrt{\xi^\top \xi}$ denotes the Euclidean norm. For $\rho>0$, the sphere of radius $\rho$ centered around the origin is denoted by $\mathcal{B}_\rho =\{\xi\in \R^n| \ |\xi|\leq \rho\}$. For a continuously differentiable function $f:\mathbb{R}^n \rightarrow \mathbb{R}$, the gradient is denoted by $\nabla f:\mathbb{R}^n \rightarrow \mathbb{R}^n$. Similarly, for $f:\mathbb{R}^n \times \mathbb{R}^m \rightarrow \mathbb{R}$ where $(\xi_1,\xi_2)\mapsto f(\xi_1,\xi_2)$ the gradient restricted to $\xi_1$ is denoted by $\nabla_{\xi_1} f:\mathbb{R}^n \times \mathbb{R}^m \rightarrow \mathbb{R}^n$. The set-valued sign function $\sign:\R \rightrightarrows \R$ satisfies
\begin{align}
    \sign(x) = \left\{ \begin{array}{rl}
       \{-1\}  & \text{if } x<0  \\
        \left[-1,1\right] & \text{if } x=0 \\
        \{1\} & \text{if } x>0.
    \end{array} \right.
\end{align}
by definition, i.e., $\sign(x)$ is uniquely defined for $x\neq 0$ and $\sign(0)$ can attain any value in the interval $[-1,1]$. 
Time is denoted through the parameter $t\in\R$.  
Additionally, for $T \geq 0$ we use 
\begin{align}
    \tau= T-t
\end{align}
to denote time in the backwards direction. The derivatives of a differentiable function $\xi:\R \rightarrow \R^n$ with respect to $t\in \R_{\geq 0}$ 
and $\tau\in \R_{\geq 0}$ are denoted by $\frac{\mathrm{d}}{\mathrm{d}t}\xi(t) = \dot{\xi}(t)$ and $\frac{\mathrm{d}}{\mathrm{d}\tau}\xi(\tau) = \mathring{\xi}(\tau)$, respectively.

\section{Problem Formulation}
\label{sec:problem}

In this section, we introduce the surveillance-evasion game of interest, and describe its solution as both a \emph{game of kind} and a \emph{game of degree}.
Consider an evader and a pursuer moving in the two-dimensional Euclidean plane.
The evader maintains a constant speed $v_e > 0$ and its position $\xi_e = [x_e,y_e]^\top \in \mathbb{R}^2$ and heading angle $\theta_e \in (-\pi,\pi]$ evolve according to the unicycle (or Dubins car) kinematic equations
\begin{align}
   \label{eq:cartesianDynamics}
   \begin{aligned}
   \dot{x}_e(t)
   &= v_e \sin \theta_e(t)\\
   \dot{y}_e(t)
   &= v_e \cos \theta_e(t)\\
   \dot{\theta}_e(t)
   &= \omega_e u_e(t).
   \end{aligned}
\end{align}
Here, $\omega_e > 0$ is the evader's finite maximum turn rate in radians per second.
The evader's control input is its normalized turn rate 
$u_e(t) \in [-1,1]$ for all $t\in \R_{\geq 0}$.
Here, $u_e(t) \in (0,1]$ corresponds to right-hand turns, $u_e(t) \in [-1, 0)$ corresponds to left-hand turns, and $u_e(t) = 0$ corresponds to no turning motion.

The pursuer similarly moves with a constant speed $v_p \geq 0$ but its position $\xi_p = [x_p,y_p]^\top \in \mathbb{R}^2$ and heading $\theta_p \in (-\pi,\pi]$ are described by the (agile) kinematic equations
\begin{align}
   \label{eq:cartesianDynamicsPursuer}
   \begin{aligned}
   \dot{x}_p(t)
   &= v_p \sin \theta_p(t)\\
   \dot{y}_p(t)
   &= v_p \cos \theta_p(t).
   \end{aligned}
\end{align}
The pursuer is agile in the sense that it is capable of instantaneously changing its heading angle $\theta_p(t)$ (i.e. it has an infinite maximum turn rate).
Thus, $\theta_p(t) \in (-\pi, \pi]$, $t\in \R_{\geq 0}$, is the pursuer's control input.

To pose the prying pedestrian surveillance-evasion game, 
we define the distance between
the evader and pursuer as 
\begin{align}
    r
    & = |\xi_e - \xi_p| = \sqrt{(x_e - x_p)^2 + (y_e - y_p)^2}, \label{eq:def_r}
\end{align}
which is a function of time $r(\cdot):\R_{\geq 0} \rightarrow \R_{\geq 0}$ since the positions $\xi_e(\cdot)$ and $\xi_p(\cdot)$ expressed in the global reference frame change with time $t\geq 0$.
The pursuer and evader initially start within a given \emph{surveillance 
distance}
$\rho > 0$ of  
each other, i.e., $r(0) \leq \rho$.
The objective of the evader is to ``escape'' from the pursuer by increasing 
the range $r(t)$ beyond the surveillance range $r(t)> \rho$  
(i.e. to achieve $r(T) = \rho$ and $\dot{r}(T) > 0$ simultaneously at some time $T > 0$).
Conversely, the objective of the pursuer is to stay within surveillance range of the evader (i.e. to keep $r(t) \leq \rho$ for all $t \geq 0$).
In the context of this problem formulation, the \emph{game of kind} (i.e.\ whether the evader can escape) and the \emph{game of degree} (i.e.\ how quickly the evader can escape) \cite[Section 2.5]{Isaacs65}, can be considered. 
While the game of kind has a straightforward solution (see Section \ref{sec:game_of_kind}, in the following), a derivation of the solution of the game of degree is more involved and is the main focus of this paper. The game of degree is formulated in Section 
\ref{sec:game_of_degree_intro}
and a solution is presented in
Section \ref{sec:solution_game_of_degree}.
We shall simplify the derivations and discussion of these solutions by using a coordinate system centered on the evader.

\subsection{Evader-Centric Coordinate System}
Let us define a coordinate system with its origin fixed on the (moving) evader (i.e., the origin is at $\xi_e$); its $y$-axis aligned with the evader's heading $\theta_e$; and, its $x$-axis orientated at an angle of $\pi/2$ radians clockwise from the positive $y$-axis (as shown in Fig.~\ref{fig:fig1}).
\begin{figure}[t!]
    \centering
    \begin{overpic}[width = 0.5\columnwidth]{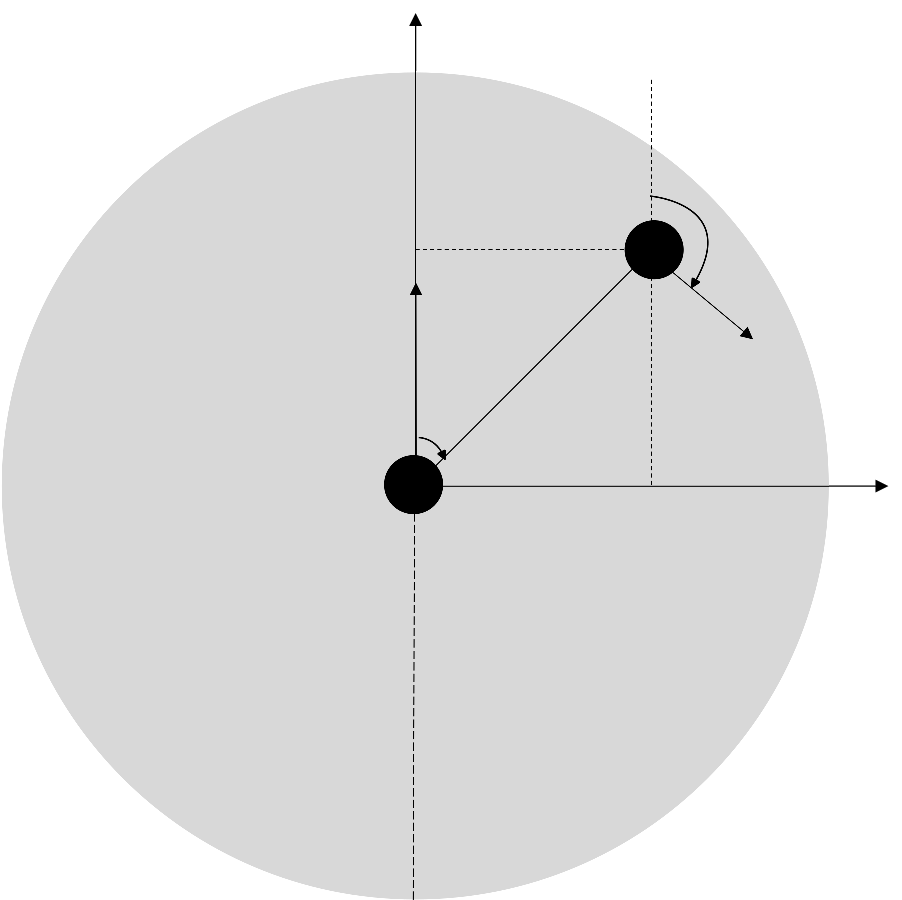}
    \put(30,45){\small{Evader}}
    \put(55,75){\small{Pursuer}}
    \put(42,20){\small{$\rho$}}
    \put(42,72){\small{$y$}}
    \put(42,62){\small{$v_e$}}
    \put(48,53){\small{$\theta$}}
    \put(57,59){\small{$r$}}
    \put(72,42){\small{$x$}}
    \put(83,64){\small{$v_p$}}
    \put(80,77){\small{$\theta_p-\theta_e$}}
    \end{overpic}
    \caption{Coordinate system attached to the evader.}
    \label{fig:fig1}
\end{figure}
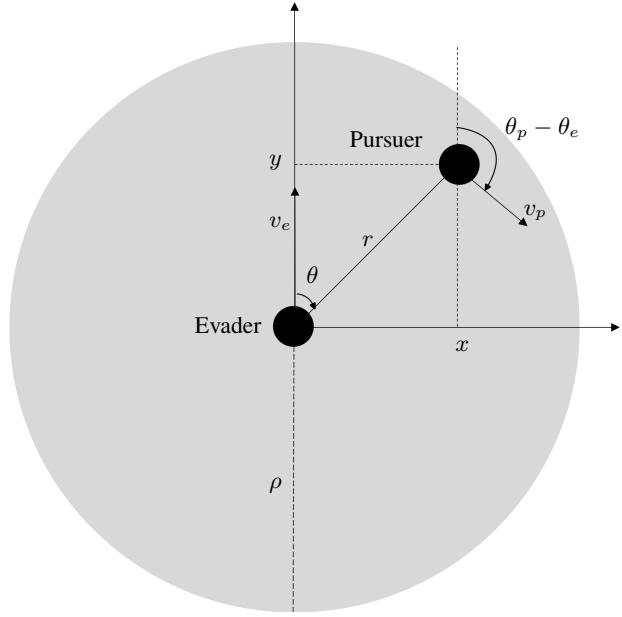
The transformation between the original (inertial) coordinates and this new evader-centric coordinate system is thus
\begin{align}
    \xi=\left[\begin{array}{c}
         x \\
         y
 \end{array} \right] = \left[\begin{array}{rr}
     \cos\theta_e & -\sin\theta_e  \\
     \sin\theta_e & \cos\theta_e 
 \end{array} \right] \left[ 
\begin{array}{c}
     x_p -x_e \\
     y_p -y_e
\end{array}
 \right] \in \mathbb{R}^2. \label{eq:coordinate_transformation_xi}
\end{align}
In addition,
since the pursuer can instantaneously change its heading angle $\theta_p(t)$,  we consider a second coordinate transformation in the pursuer's degree of freedom and define its control input to be the relative heading angle 
\begin{align}
u_p(t) = \theta_p(t) - \theta_e(t) \in (-\pi,\pi],  \label{eq:input_u_p}
\end{align}
which it can also instantaneously change.
With this choice of coordinate system and pursuer control, the state of the differential game is the relative position of the pursuer, i.e. $\xi(t)$, which via \eqref{eq:cartesianDynamics} and \eqref{eq:cartesianDynamicsPursuer} evolves according to the dynamics 
\begin{align}
    \label{eq:dynamics}
    \dot{\xi}(t)
    &= f(\xi(t), u_e(t), u_p(t))
\end{align}
for $t \geq 0$, where
\begin{align}
    f(\xi, u_e, u_p)
    &=
    \begin{bmatrix}
        -\omega_e y u_e + v_p \sin u_p\\
        \omega_e x u_e - v_e + v_p \cos u_p
    \end{bmatrix}. \label{eq:f_dynamics}
\end{align}
A derivation of the dynamics \eqref{eq:dynamics} is provided in Appendix \ref{sec:coordinate_transformation} for completeness. We
note that, according to \eqref{eq:def_r}, the distance between the evader and the pursuer in the $\xi$-coordinates \eqref{eq:coordinate_transformation_xi} simplifies to 
\begin{align}
    r &= |\xi| = \sqrt{x^2 + y^2}. \label{eq:def_r2}
\end{align}
Alternatively, the components of the state $\xi \in \R^2$ can be represented in polar coordinates with distance $r\geq 0$ and angle $\theta \in (-\pi,\pi]$ through the coordinate transformation
\begin{align}
    x=r \sin(\theta) \quad \text{ and } \quad y=r \cos(\theta). \label{eq:sin_cos_transform}
\end{align}
Finally, to describe the motion of the pursuer either forward or backward in time, we shall denote solutions of \eqref{eq:dynamics} with respect to initial condition $\xi_0\in \R^2$, at initial (forward) time $t_0 \in \R_{\geq 0}$ or at initial (backward) time $ \tau_0 $, for inputs $u_e,u_{e_\tau}:\R_{\geq 0} \rightarrow [-1,1]$, $u_p,u_{p_\tau}:\R_{\geq 0} \rightarrow (-\pi,\pi]$, as
\begin{align}
\left[ \begin{smallmatrix}
        x(\cdot) \\
        y(\cdot)
    \end{smallmatrix} \right] = 
    \xi(\cdot) = \xi_t(\cdot;t_0,\xi_0,u_e(\cdot),u_p(\cdot)) \qquad \text{ and } \qquad 
    \left[ \begin{smallmatrix}
        x_\tau(\cdot) \\
        y_\tau(\cdot)
    \end{smallmatrix} \right] =
    \xi_\tau(\cdot) = \xi_\tau(\cdot;\tau_0,\xi_0,u_{e_\tau}(\cdot),u_{p_\tau}(\cdot)), \label{eq:solution_notation}
\end{align}
respectively.

\begin{remark} \label{rem:continuity_of_solutions}
Note that the function $f$ given by \eqref{eq:f_dynamics} is locally Lipschitz continuous with respect to $(\xi,u_e,u_p)$, and thus for piecewise Lipschitz continuous functions $u_e(\cdot)$, $u_p(\cdot)$, solutions $\xi(\cdot)$, $\xi_\tau(\cdot)$ in \eqref{eq:solution_notation} are absolutely continuous functions and satisfy \eqref{eq:dynamics} for almost all $t,\tau\in [0,\infty)$. 
\hfill $\diamond$
\end{remark}

\subsection{Game of Kind} \label{sec:game_of_kind}

The solution to the game of kind, i.e., whether the evader can escape the pursuer, can be characterized solely based on the pursuer's speed $v_p$ and the evader's speeds $v_e$.
Specifically, when the pursuer is faster than the evader (i.e., when $v_p > v_e$), the evader will never be able to escape provided that the pursuer employs an appropriate strategy.
To see this, consider a point $\xi \in \R^2$ with $|\xi|=\rho$ and polar coordinates $x=\rho \sin(\theta)$, $y=\rho \cos(\theta)$ 
for $\theta \in (-\pi,\pi]$ (cf.\ \eqref{eq:sin_cos_transform}).
Then, 
the input $u_p=\pi+\theta$ guarantees
\begin{align}
\begin{split}
      f(\xi,u_e,u_p)^\top \xi&=
    (-\omega_e y u_e + v_p \sin u_p)x+
(\omega_e x u_e - v_e + v_p \cos u_p)y \\
 &= v_p (x \sin(u_p) +y\cos(u_p) ) - y v_e \\
 &= v_p (\rho \sin(\theta) \sin(\pi+\theta) +\rho \cos(\theta)\cos(\pi+\theta) ) - \rho \cos(\theta) v_e \\
 & \leq \rho (- v_p  +  v_e)  < 0.
 \end{split} \label{eq:decrease_game_kind}
\end{align}
In other words, $f(\xi,u_e,u_p)$ points inside the sphere $\mathcal{B}_\rho$ as illustrated in Figure \ref{fig:increase_decrease}.
\begin{figure}[t!]
    \centering
    \begin{overpic}[width = 0.4\columnwidth]{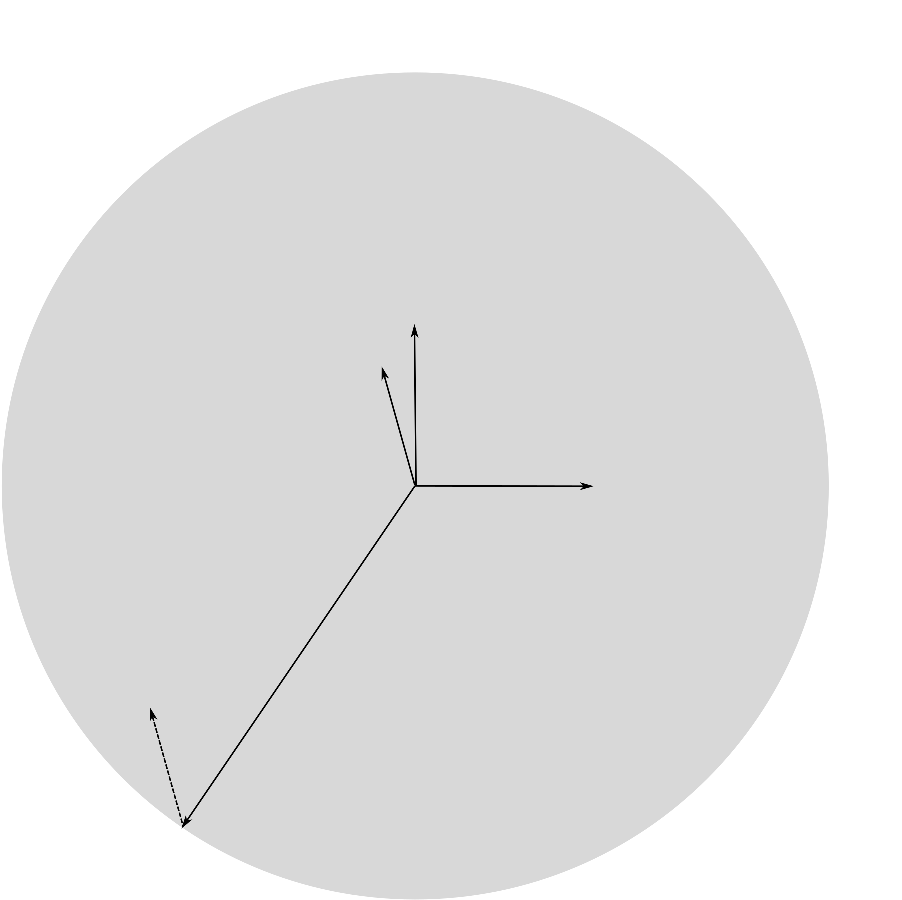}
    \put(64,42){\small{$x$}}
    \put(48,60){\small{$y$}}
    \put(32,20){\small{$\xi$}}
    \put(18,50){\small{$f(\xi,u_e,u_p)$}}
    \put(-7,13){\small{$f(\xi,u_e,u_p)$}}
    \end{overpic}
    \hspace{0.3cm}
    \begin{overpic}[width = 0.4\columnwidth]{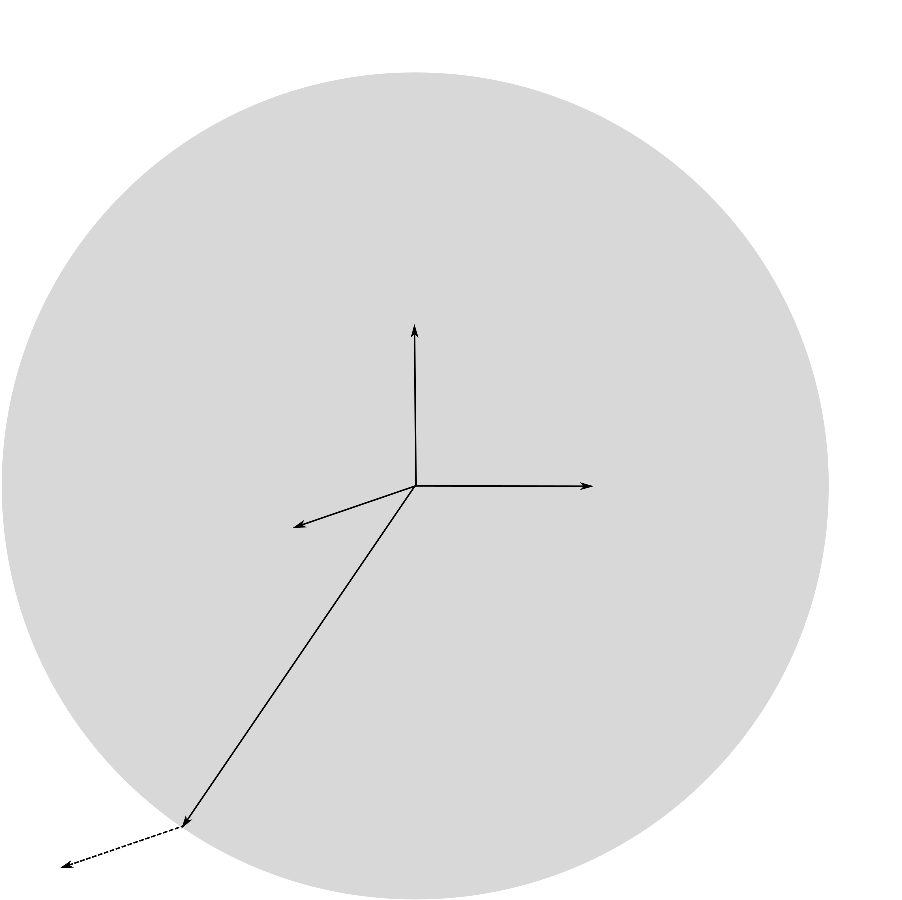}
    \put(64,42){\small{$x$}}
    \put(48,60){\small{$y$}}
    \put(32,20){\small{$\xi$}}
    \put(12,45){\small{$f(\xi,u_e,u_p)$}}
    \put(-10,10){\small{$f(\xi,u_e,u_p)$}}
    \end{overpic}
    \caption{Illustration of the decrease condition in \eqref{eq:decrease_game_kind}. On the left, $\xi$ and $f(\xi,u_e,u_p)$ point in opposite directions and thus $f(\xi,u_e,u_p)^\top \xi < 0$. On the right, $\xi$ and $f(\xi,u_e,u_p)$ point in the same direction and thus  $f(\xi,u_e,u_p)^\top \xi > 0$.}
    \label{fig:increase_decrease}
\end{figure}
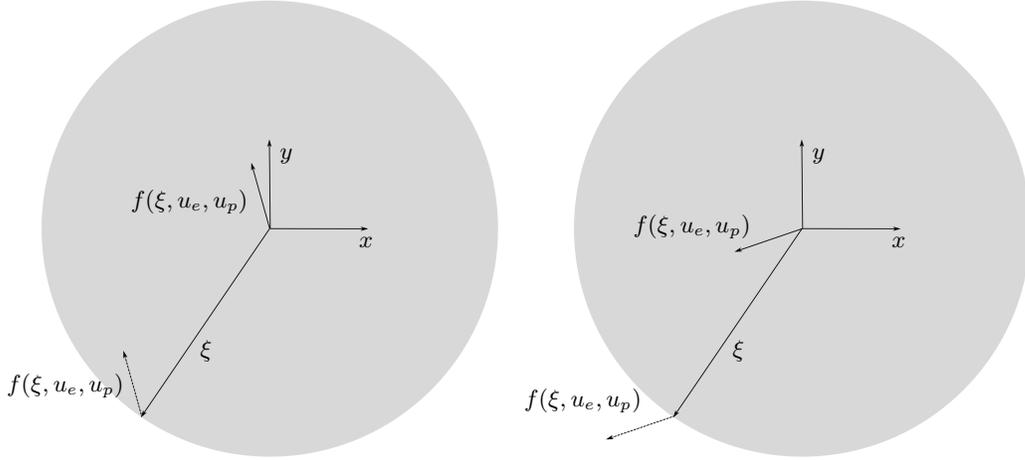
Thus, a faster pursuer can keep a slower evader under surveillance indefinitely by simply selecting a piecewise-constant feedback law $u_p(t)=\pi+\bar{\theta}$ with $\bar{\theta}$ updated to $\theta(t)$ whenever the (relative) position $\xi(t)$  
nears the boundary of the surveillance region (i.e., whenever $|\xi(t)|=\rho$). 
This pursuer strategy guarantees that $|\xi(t)|\leq \rho$ for all $t \in \R_{\geq 0}$.
When the pursuer and evader have the same speed (i.e. when $v_e = v_p$), the same strategy guarantees that $\tfrac{\mathrm{d}}{\mathrm{d}t}|\xi(t)| \leq 0$ whenever $|\xi(t)| = \rho$ and thus, the same conclusion can be drawn.
When the pursuer is slower than the evader (i.e. when $v_p < v_e$), the evader can escape by simply maintaining its initial heading (i.e. by selecting $u_e(t) = 0$ for all $t \geq 0$). In this case, the maximum time that the evader will take to escape is $T = 2\rho/(v_e - v_p)$ and corresponds to the pursuer initially being in front of the evader with $r(0) = \rho$ and selecting $\theta_p(t)$ so as to always point in the same direction as the evader. 

\begin{remark} \label{eq:differentiability_issue}
    Note that in the case that $\xi(\cdot)$ is continuously differentiable in $t\in \R_{\geq 0}$, \eqref{eq:decrease_game_kind} implies that
\begin{align}
    \tfrac{1}{2} \tfrac{\mathrm{d}}{\mathrm{d}t}|\xi(t)|^2 = f(\xi(t),u_e(t),u_p(t))^\top \xi(t) \label{eq:d_dt_norm_xi}
\end{align}
is negative, i.e., $|\xi(t)|$ is strictly decreasing in a neighborhood around $t\in \R_{\geq 0}$. To deal with the fact that $\xi(\cdot)$ as an absolutely continuous function is not necessarily continuously differentiable, we have avoided the derivative in the derivations \eqref{eq:decrease_game_kind} and instead rely on a pointwise condition with illustration in Figure \ref{fig:increase_decrease}.  
For simplicity of notation, we will use the derivative in the following for all $t\in \R_{\geq 0}$ despite the fact that the derivative might not be well-defined at isolated times $t\in \R_{\geq 0}$.
\hfill $\diamond$
\end{remark}

\subsection{Game of Degree} \label{sec:game_of_degree_intro}

In the game of degree, the evader seeks to minimize the time it takes to escape whilst the pursuer seeks to maximize it.
Our analysis of the game of kind suggests that the game of degree is degenerate when the pursuer is at least as fast as the evader (i.e. when $v_p \geq v_e$) since a slower evader can never escape in finite time provided that the faster pursuer employs a sensible strategy (such as that we previously identified).
In the remainder of this paper, we shall therefore consider the game of degree under the assumption that the evader is faster than the pursuer (i.e.\ when $v_p < v_e$).
We shall specifically treat the game of degree as a two-player zero-sum differential game (cf.\ \cite{Isaacs65} or \cite[Chapter 8]{Basar1999}).

In the game of degree, the time $T > 0$ at which the evader escapes by achieving 
$|\xi(T)| = \rho$ and $\frac{\mathrm{d}}{\mathrm{d}t}|\xi(T)|^2 >0$
is the game's payoff.
The value function defining the game of degree (and the optimal time-to-go in the game) from a given state $\xi_0=[x_0,y_0]^\top$ is thus
\begin{align}
\begin{split}
V(\xi_0) &= \min_{u_e} \max_{u_p} \int_0^T 1 \, \mathrm{d}t 
\end{split} \label{eq:problem_V}
\end{align}
where the minimization is over evader control functions 
$u_e : [0,T] \rightarrow [-1,1]$ and the maximization is over pursuer control functions $u_p : [0,T] \rightarrow (-\pi,\pi]$. 
Accordingly, the game of degree can be summarized as the following problem.

\begin{problem}[Game of Degree] \label{prob:game_of_degree}
Consider the dynamics \eqref{eq:dynamics} with state $\xi\in \R^2$, inputs $u_e\in [-1,1]$ and $u_p\in (-\pi,\pi]$, and defined through parameters $\omega_e,v_e,v_p\in \R_{>0}$ with $v_p<v_e$.
Additionally, consider running costs and terminal costs defined as
\begin{align}
    L(\xi,u_e,u_p)= 1 \quad \text{and} \quad G(\xi) = 0, \label{eq:defLG}
\end{align}
respectively.
The game of degree with surveillance radius $\rho>0$ is the optimization problem
\begin{align}
\begin{split}
V(\xi_0) = \min_{u_e} \max_{u_p} & \int_0^T L(\xi(t),u_e(t),u_p(t)) \, \mathrm{d}t + G(\xi(T)) \\
\text{subject to } \quad & \dot{\xi}(t) = f(\xi(t), u_e(t), u_p(t)), \quad \xi(0) = \xi_0, \\
& u_e(t) \in [-1,1], \quad t \in [0,T],\\
& u_p(t) \in \mathbb{R}, \quad t \in [0,T],\\
& |\xi(0)| \leq   \rho, \quad |\xi(T)| = \rho, \\
& f(\xi(T),u_e(T),u_p(T))^\top \xi(T) > 0.
\end{split} \label{eq:problem_constraints}
\end{align}

\vspace{-0.8cm}

\flushright $\triangleleft$
\end{problem}

We denote optimal evader and pursuer strategies solving \eqref{eq:problem_constraints} as $u_e^*:\R_{\geq 0} \rightarrow [-1,1]$ and $u_p^*:\R_{\geq 0} \rightarrow (-\pi,\pi]$, respectively, in the following. 
Similarly, the minimal time at which the game ends is denoted by $T^*\in \R_{\geq 0}$.
Under the assumption that $\xi(\cdot)$ is continuously differentiable, the condition $ f(\xi(T),u_e(T),u_p(T))^\top \xi(T)>0$ can be replaced by $\tfrac{\mathrm{d}}{\mathrm{d}t}|\xi(T)|^2 > 0$. However, here we use the slightly more general notation.

\begin{remark}
    To avoid discontinuities in the input $u_p$ when the angle crosses $\pi$, we use $u_p : \R_{\geq 0} \rightarrow \mathbb{R}$ in the following. The original input  
    can always be recovered through $[(u_p-\pi)\mod 2\pi]+\pi \in (-\pi,\pi]$. 
    \hfill $\diamond$
\end{remark}

A solution of Problem \ref{prob:game_of_degree} is derived in the following section. 
As it turns out, the solution of the game of degree is more difficult to characterize for certain speed ratios $\frac{v_p}{v_e} \in [0,1)$.
As illustrated in Figure \ref{fig:solution_game_degree_intro_new}, depending on the speed ratio $\frac{v_p}{v_e} \in [0,1)$, the solution to the game of degree can consist of two or three regions.
\begin{figure}[htb!]
    \centering
    \begin{overpic}[width = 0.45\columnwidth]{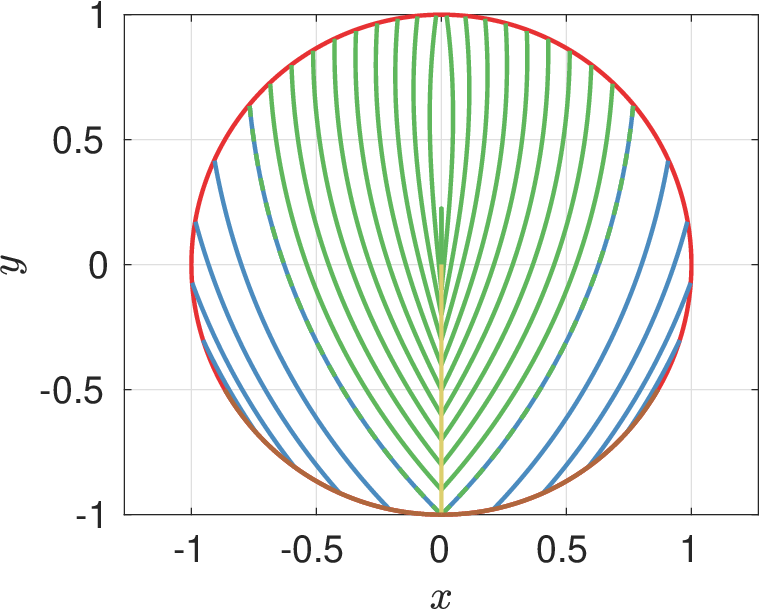}
    \put(30,40){\large{(A)}}
    \put(40,50){\large{(B)}}
    \end{overpic} 
    \hspace{0.4cm}
    \begin{overpic}[width = 0.45\columnwidth]{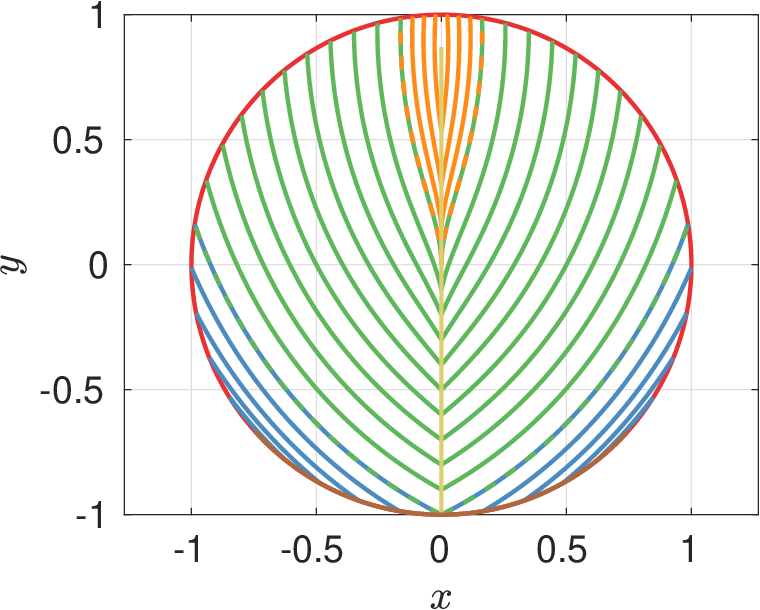}
    \put(30,25){\large{(A)}}
    \put(40,50){\large{(B)}}
    \put(55,70){\large{(C)}}
    \end{overpic}
    \caption{Trajectories corresponding to solutions of the game of degree for different parameter selections. On the left, the parameters $v_p=1$, $v_e=2$, $\rho=1$ and $\omega_e=1$ are used. On the right, $v_e$ is replaced by $v_e=1.5$.}
    \label{fig:solution_game_degree_intro_new}
\end{figure}
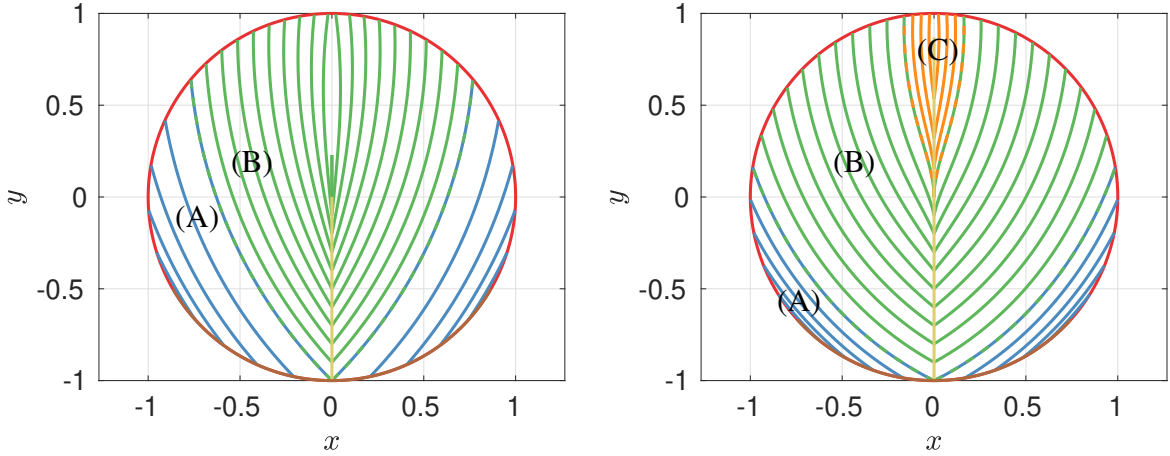
As will be shown in Section \ref{sec:weird_part}, for $\frac{v_p}{v_e}\leq \frac{1}{2}$, the solution can be divided in two regions, denoted (A) and (B) in Figure \ref{fig:solution_game_degree_intro_new} on the left. In the case that $\frac{v_p}{v_e}> \frac{1}{2}$, three regions (A), (B), and (C) are necessary as indicated in Figure \ref{fig:solution_game_degree_intro_new} on the right. 
In this paper, we provide a detailed derivation of the solution to the game of degree in all regions (A), (B) and (C).

\section{Game of Degree Solution} \label{sec:solution_game_of_degree}

In this section, we derive a solution of the game of degree specified in Problem \ref{prob:game_of_degree}. 
The solution is divided in several components corresponding to the domains shown in Figure \ref{fig:solution_game_degree_intro_new}.

\subsection{Preparatory Results and Definitions}

We start the analysis of Problem \ref{prob:game_of_degree} by introducing definitions and by deriving preparatory results. The definitions follow the notations in \cite{Lewin2012}. 
The \emph{game set} of the game of degree \eqref{prob:game_of_degree}, i.e., the domain of the state space where the game is defined, is given by
\begin{align}
\begin{split}
    \mathcal{S} &= \{ \xi \in \R^2 | \ x^2 + y^2 \leq \rho^2 \} \subset \R^2.
\end{split} \label{eq:S}
\end{align}
The \emph{target set} of the game, i.e., the set where the game can potentially end, is given by
\begin{align}
    \mathcal{C} \hspace{-0.05cm} &= \hspace{-0.05cm} \{ \xi \in \mathcal{S} |  x^2 + y^2 = \rho^2 \} \subset \mathcal{S}.  \label{eq:C}
\end{align}
According to \eqref{eq:problem_constraints} the game of degree ends if the conditions $|\xi(T)| = \rho$ (i.e., $\xi(T) \in \mathcal{C}$) and 
$f(\xi(T),u_e(T),u_p(T))^\top \xi(T) > 0$
are satisfied. Under the assumptions of optimal play, the game can only be terminated by the evader by leaving the game set $\mathcal{S}$. 
To take the condition 
$f(\xi(T),u_e(T),u_p(T))^\top \xi(T) > 0$
on the set $\mathcal{C}$ into account, the target set is partitioned into the \emph{usable part} (UP), the \emph{nonusable part} (NUP) and the \emph{boundary of the usable part} (BUP).
We establish the meaning of, and expressions for, these partitions in terms of the surveillance radius $\rho$ and the speed ratio $\tfrac{v_p}{v_e}$ in the following lemma (using the approach of \cite[Ch. 5.2.1]{Lewin2012}).

\begin{lemma}[Usable and Nonusable parts] \label{lem:usable_part}
Consider Problem \ref{prob:game_of_degree}. 
Under optimal play (of the evader and the pursuer), the game of degree \eqref{eq:problem_constraints} can only end on the usable part defined through the set
\begin{align}
    \text{UP} &= \left\{ \xi \in \mathcal{C} | y < - \rho \tfrac{v_p}{v_e}  \right\} \label{eq:UP1}.
    \end{align}
    Conversely, under optimal play, the game cannot end on the nonusable part and the boundary of the usable part given by
    \begin{align}
    \text{NUP} &= \left\{ \xi \in \mathcal{C} | y > - \rho \tfrac{v_p}{v_e}  \right\} \qquad \text{and} \qquad     \text{BUP} = \left\{ \xi \in \mathcal{C} | y = - \rho \tfrac{v_p}{v_e}  \right\}, \label{eq:BUP1}
\end{align}
respectively.
\hfill  $\lrcorner$
\end{lemma}

\begin{proof}
For the game of degree to end, the conditions 
$|\xi| = \rho$ and  $\min_{u_e\in [-1,1]} \max_{u_p \in\R} f(\xi,u_e,u_p)^\top \xi>0$ need to be satisfied. 
To examine when these hold, consider any point $\xi$ on the boundary with angle $\theta_{\xi} \in (-\pi,\pi]$ such that
$x=\rho \sin(\theta_{\xi})$ and $y = \rho  \cos(\theta_{\xi})$ (using the representation \eqref{eq:sin_cos_transform}).
It holds that
\begin{align}
\begin{split}
    f(\xi,u_e,u_p)^\top \xi
 &=   (-\omega_e y u_e + v_p \sin u_p)x+ (\omega_e x u_e - v_e + v_p \cos u_p)y \\
 &= v_p (x \sin(u_p) +y\cos(u_p) ) - y v_e \\
 &= v_p (\rho \sin(\theta_\xi) \sin(u_p) +\rho \cos(\theta_\xi)\cos(u_p) ) - \rho \cos(\theta_\xi) v_e.
\end{split} \label{eq:UP_calculation}
\end{align}
The last expression is notably independent of the evader's input $u_e$.
The best the pursuer can do to keep the game going at this point (i.e., to ensure that 
$f(\xi,u_e,u_p)^\top \xi\leq 0$)
is to select
\begin{align}
    u_p = \pi +\theta_\xi, \label{eq:best_up}
\end{align}
which further leads to
\begin{align}
\begin{split}
     f(\xi,u_e,u_p)^\top \xi& =  v_p (\rho \sin(\theta_\xi) \sin(u_p) +\rho \cos(\theta_\xi)\cos(u_p) ) - \rho \cos(\theta_\xi) v_e \\ 
    &\geq v_p (\rho \sin(\theta_\xi) \sin(\theta_\xi+\pi) +\rho \cos(\theta_\xi)\cos(\theta_\xi+\pi) ) - \rho \cos(\theta_\xi) v_e \\
    &= -v_p(\rho \sin^2(\theta_\xi)  +\rho \cos^2(\theta_\xi)) - \rho \cos(\theta_\xi) v_e \\
    &= -\rho v_p  - \rho \cos(\theta_\xi) v_e = -\rho v_p  - y(0) v_e.
\end{split} \label{eq:proof_equations_UB}
\end{align}
Thus, for $y\leq -\rho \frac{v_p}{v_e}$ it holds that 
$f(\xi,u_e,\pi+\theta_\xi)^\top \xi>0$
which leads to the characterization of the usable part. Similarly, the conditions 
$f(\xi,u_e,\pi+\theta_\xi)^\top \xi<0$
and 
$f(\xi,u_e,\pi+\theta_\xi)^\top \xi=0$ for all $u_e\in [-1,1]$
lead to the nonusable and the boundary of the usable part. 
\end{proof}

Consider Figure \ref{fig:increase_decrease}. The usable part, the boundary of the usable part, and the nonusable part divide the set $\mathcal{C}$ into regions where $f(\xi,u_e,\pi+\theta_\xi)$, independent of the selection $u_e\in [-1,1]$, points outside of $\mathcal{S}$, is tangential to the boundary of $\mathcal{S}$, and points inside $\mathcal{S}$, respectively.
Here, as in the proof of Lemma \ref{lem:usable_part},
$\theta_{\xi} \in (-\pi,\pi]$ is defined through the conditions
$x=\rho \sin(\theta_{\xi})$ and $y = \rho  \cos(\theta_{\xi})$.

\begin{remark} \label{rem:u_p_star}
Note that for
$x=r \sin(\theta)$, $y=r \cos(\theta)$, for $r>0$, the angle $\pi+ \theta$ corresponds to the pursuer pointing to the evader. This is not only valid on $\mathcal{C}$, as used in  \eqref{eq:best_up}, but for all $\xi \in \mathcal{S}$. 
\hfill $\diamond$
\end{remark}

\begin{remark} \label{rem:u_e_independence}
Since the evader's input disappears in  \eqref{eq:proof_equations_UB}, it can be shown that the pursuer can guarantee that $f(\xi,u_e,u_p)^\top \xi = 0$ for all $\xi \in \mathcal{C}\cap \text{UP}$ by appropriately selecting $u_p\in (-\pi,\pi]$ independent of the selection $u_e\in [-1,1]$. Thus, the pursuer can keep a constant distance to the evader until the usable part is reached.  \hfill $\diamond$
\end{remark}

The next lemma describes a candidate optimal evader strategy.
We shall use this result later to prove that the evader's optimal strategy is, indeed, given by $u_e^*(\xi) \in -\sign(x)$.

\begin{lemma} \label{lem:y_decrease}
Consider Problem \ref{prob:game_of_degree} and assume that the evader uses the strategy
\begin{align}
u_e^\#(t)    \in -\sign(x(t)), \qquad \forall \ \xi(t)  = \left[\begin{smallmatrix}
    x(t) \\
    y(t)
\end{smallmatrix} \right]\in \mathcal{S}. \label{eq:evader_opt_strat}
\end{align}
Moreover, let $u_p(\cdot):\R_{\geq 0} \rightarrow (-\pi,\pi]$ denote an arbitrary piecewise Lipschitz continuous strategy of the pursuer. 
Then, independent of the pursuer's strategy, the $y$-component of the solution
$y(t)=y_t(t;\xi_0,u_e^\#(t),u_p(t))$ is strictly monotonically decreasing and 
\begin{align}
    y(t)-y(0) \leq t\cdot (v_p-v_e)
\end{align}
for all $t\in \R_{\geq 0}$ such that $\xi(t;\xi_0,u_e^\#(t),u_p(t)) \in \mathcal{S}$.
\hfill $\lrcorner$
\end{lemma}

\begin{proof}
Since $u_e^\#(\cdot)$ and $u_p(\cdot)$ are piecewise Lipschitz continuous by assumption, for almost all $t\in \R_{\geq 0}$, it holds that
    \begin{align*}
        \tfrac{\mathrm{d}}{\mathrm{d}t}y(t) &= f(\xi(t),u_e^\#(t),u_p(t))^\top \left[\begin{array}{c}
            0 \\
            1
        \end{array} \right] = \omega_e x(t) u_e^\#(t) - v_e + v_p \cos (u_p(t)) \\
        &=-\omega_e |x(t)| -v_e + v_p \cos (u_p(t)) \leq  -v_e + v_p  
    \end{align*}
    which shows the assertion and completes the proof.
\end{proof}

As a next preliminary result we show that the game of degree is symmetric with respect to the $y$-axis. To this end, consider the coordinate transformation 
\begin{align}
    \tilde{\xi} = \left[ \begin{array}{r}
         -x \\
         y
    \end{array}\right], \label{eq:xi_tilde}
\end{align}
which, in combination with the dynamics \eqref{eq:dynamics}, implies that
\begin{align}
\left[ \begin{array}{r}
         -\dot{\tilde{x}} \\
         \dot{\tilde{y}} 
    \end{array} \right] =
\left[ \begin{array}{c}
         \dot{x} \\
         \dot{y} 
    \end{array} \right] =         
    \begin{bmatrix}
        -\omega_e \tilde{y} u_e + v_p \sin u_p\\
        -\omega_e \tilde{x} u_e - v_e + v_p \cos u_p
    \end{bmatrix}.
\end{align}
With the coordinate transformations $\tilde{u}_p=-u_p$ and $\tilde{u}_e=-u_e$ the last expression can be further rewritten as
\begin{align}
\dot{\tilde{\xi}}=
\left[ \begin{array}{c}
         \dot{\tilde{x}} \\
         \dot{\tilde{y}} 
    \end{array} \right]  =         
    \begin{bmatrix}
        -(-\omega_e \tilde{y} u_e + v_p \sin u_p)\\
        -\omega_e \tilde{x} u_e - v_e + v_p \cos u_p
    \end{bmatrix} =
    \begin{bmatrix}
        -\omega_e \tilde{y} (-u_e) + v_p \sin (-u_p))\\
        \omega_e \tilde{x} (-u_e) - v_e + v_p \cos (-u_p) 
    \end{bmatrix} 
    = f(\tilde{\xi},\tilde{u}_e,\tilde{u}_p).
\end{align}
Thus, the original dynamics \eqref{eq:dynamics} is recovered, proving the following result.

\begin{corollary}[Symmetry with respect to $y$-axis] \label{cor:symmetry}
Consider Problem \ref{prob:game_of_degree}. If $u_e^*(\cdot)$ and $u_p^*(\cdot)$ are optimal for $\xi_0\in \mathcal{S}$, then $-u_e^*(\cdot)$ and $-u_p^*(\cdot)$ are optimal for $\tilde{\xi}_0\in \mathcal{S}$ in \eqref{eq:xi_tilde}. 
\hfill $\lrcorner$
\end{corollary}

With this corollary, the optimal solution of the game on the negative $y$-axis can be derived.

\begin{lemma} \label{lem:strategy_on_y_axis}
Consider Problem \ref{prob:game_of_degree} and let $\xi_0\in \{\xi \in \mathcal{S}| x=0, \ y< 0 \}$. Then the optimal evader and pursuer strategies satisfy 
\begin{align}
    u_e^*(t)=0  \quad \text{and} \quad u_p^*(t)=0 \label{eq:optimal_input_y_axis}
\end{align}
for all $t\in \R_{\geq 0}$ such that $\xi_t(t;\xi_0,0,0) \in \mathcal{S}$. Moreover, for $\xi_0\in \{\xi \in \mathcal{S}| x=0, \ y< 0 \}$ optimal solutions satisfy
\begin{align}
    \xi(t;\xi_0,0,0) = 
    \left[ \begin{array}{c}
         0 \\
         \frac{t}{v_p-v_e} +y_0 
    \end{array}\right] \label{eq:solution_y_axis}
\end{align}
for $t\in [0,\frac{\rho+y_0}{v_e-v_p}]$.
\hfill $\lrcorner$
\end{lemma}

Since the proof is not insightful, it is only reported in Appendix \ref{ap:auxilliary_results} for completeness.

\subsection{Candidate Optimal Strategies, the Hamiltonian and the Adjoint Equation}

Our preparatory results in the previous section have defined the usable part of the game of degree and derived optimal solutions for trajectories starting on the negative $y-$axis. 
In this section, we derive additional conditions on the optimal inputs $u_e^*$ and $u_p^*$ relying on the Hamiltonian function and on adjoint (adjoint) variables and the adjoint equation.
In the derivation, we focus on regular parts of the game set where the following properties are satisfied and are used in the derivation of optimal trajectories of the game.

\begin{definition}[\mbox{Regular Points \& Parts, \cite[Ch. 5.4.4]{Lewin2012}}] \label{def:regular_point}
    Consider Problem \ref{prob:game_of_degree} with game set $\mathcal{S}$ in \eqref{eq:UP_calculation}.
    A point $\xi \in \mathcal{S}$ is called regular, if the optimal value function $V:\mathcal{S} \rightarrow \R$ defined in \eqref{eq:problem_V} is continuously differentiable in $\xi$. If $\xi_t(t;\xi_0,u_e^*(t),u_p^*(t))$ is regular for all $t\in [t_1,t_2]$, $t_1,t_2\in \R_{\geq 0}$, then $\xi_t(\cdot;\xi_0,u_e^*(\cdot),u_p^*(\cdot)):[t_1,t_2]\rightarrow \mathcal{S}$ is called regular part of an optimal trajectory.
    \hfill 
    $\lrcorner$
\end{definition}

For a regular point $\xi$ the gradient of the optimal value function is denoted by $\nabla V(\xi)$. To shorten expressions in the following, we adopt the notation of adjoint variables $\bp \in \R^2$ and write 
\begin{align}
 \nabla V(\xi) = 
  \left[ \begin{array}{c}
       V_x(\xi) \\
       V_y(\xi)  
\end{array} \right] 
=\left[ \begin{array}{c}
       p_x \\
       p_y   
\end{array} \right] = \bp.
\end{align}
The definition of the gradient of the optimal value function allows us to define the Hamiltonian function
\begin{align}
\label{eq:Ham}
H(\xi,\bp,u_e, u_p) =  \bp^{\top} f(\xi, u_e, u_p) + L(\xi,u_e,u_p)
\end{align}
which satisfies 
\begin{align}
\label{eq:Ham_cond1}
\min_{u_e \in [-1,1]} H(\xi,\bp,u_e, u_p^*) \leq H(\xi,\bp,u_e^*, u_p^*) \leq \max_{u_p \in (-\pi,\pi]} H(\xi,\bp,u_e^*, u_p),
\end{align}
and
\begin{align}
    H(\xi(t),\nabla V(\xi(t)),u_e^*(t),u_p^*(t)) = H(\xi(t),\bp(t),u_e^*(t),u_p^*(t))=0  \label{eq:H_0_condition}
\end{align}
for all $t\in \R_{\geq 0}$ on regular parts (see \cite[Theorem 5.5.1]{Lewin2012}). Before we use this condition, we first derive a relation between $\nabla V(\xi)=\bp$ and the optimal inputs $u_e^*$ and $u_p^*$.

\begin{lemma}[Candidate Optimal Strategies] \label{lem:opt_control_regular_E} 
Consider the game of degree \eqref{eq:problem_constraints} defined in Problem \ref{prob:game_of_degree}. On regular parts the optimal strategy of the evader satisfies
    \begin{align}
    u_e^*& \in  \sign(p_x y-p_y  x) \label{eq:switch_function}
    \end{align}
    and the optimal strategy of the pursuer, $u_p^* \in \R$, satisfies 
    \begin{align}
        \bp = c \left[ \begin{array}{c}
         \sin(u_p^*)  \\
         \cos(u_p^*) 
    \end{array} \right] \label{eq:opt_u_star}
    \end{align}
    for $c\in \R$. \hfill $\lrcorner$    
\end{lemma}

\begin{proof}
We use condition \eqref{eq:Ham_cond1} with $H$ defined in \eqref{eq:Ham} to find optimal inputs $u_e$ and $u_p$. In particular, using the definition of $f$ in \eqref{eq:f_dynamics} and realizing that $L$ in \eqref{eq:defLG} is independent of $u_e$ and $u_p$, we consider the optimization problem
\begin{align*}
    &\min_{u_e \in [-1,1]} \max_{u_p \in \R}    \big[ p_x  (-\omega_e y u_e  +  v_p \sin u_p)   + p_y (\omega_e x u_e - v_e  +  v_p \cos u_p) \big] \\
     & \quad =  - p_y   v_e +  \min_{u_e \in [-1,1]}  
    \left[   (p_y x -p_x y)  \omega_e u_e \right]  + \max_{u_p \in \R} \left[ (p_x  \sin u_p  + p_y \cos u_p) v_p      \right].
\end{align*}
From this expression, the assertion follows immediately.
\end{proof}

As a next step, we focus on the adjoint equation
\begin{align}
    \dot{\bp} = -\tfrac{\mathrm{d}}{\mathrm{d} \xi} H(\xi,\bp,u_e^*,u_p^*),
\end{align}
which is satisfied at regular points (see \cite[Thm. 5.6.1]{Lewin2012}). 
Using the definition of the Hamiltonian \eqref{eq:Ham}, the adjoint equation can thus be derived as
\begin{align}
\begin{split}
    \dot{\bp} &=  -\tfrac{\mathrm{d}}{\mathrm{d} \xi} \left( p_x( -\omega_e y u_e + v_p \sin u_p) +   p_y (\omega_e x u_e - v_e + v_p \cos u_p) +1 \right) \\
    &=  - \left[\begin{array}{c}
        p_y \omega_e  u_e^*    \\
        -p_x \omega_e u_e^*  
    \end{array}\right] =\omega_e u_e^* \left[  \begin{array}{rr}
          0 & -1  \\
         1 & 0  \\ 
    \end{array}  \right] \left[  \begin{array}{r}
        p_x  \\
        p_y   
    \end{array} \right]=\omega_e u_e^* \left[  \begin{array}{rr}
          0 & -1  \\
         1 & 0  \\ 
    \end{array}  \right] \bp.  
\end{split}\label{eq:adjoint_equation}
\end{align}
With respect to the time $\tau = T - t$ $(T>0)$, the adjoint equation can be equivalently written as
\begin{align}
\mathring{\bp}_{\tau}  =  -\omega_e u_{e_\tau}^* \left[ \begin{array}{rr}
       0  & -1  \\
       1  & 0 
    \end{array} \right] \left[  \begin{array}{r}
        p_{x_\tau}  \\
        p_{y_\tau}   
    \end{array} \right]= -\omega_e u_{e_\tau}^* \left[ \begin{array}{rr}
       0  & -1  \\
       1  & 0 
    \end{array} \right] \bp_\tau. \label{eq:dot_p_equation_tau}
\end{align}

For later deciding whether to use $u_e=-1$ or $u_e=1$ according to \eqref{eq:switch_function}, let us define the switching functions $\sigma,\sigma_\tau:\R_{\geq 0} \rightarrow \R$ as
\begin{align}
    \sigma(t) = p_{x}(t) y(t) - p_{y}(t) x(t), \qquad     \sigma_\tau(\tau) = p_{x_\tau}(\tau) y_\tau(\tau) - p_{y_\tau}(\tau) x_\tau(\tau) \label{eq:sigma_function}
\end{align}
and consider their derivatives with respect to forward time $t$ and backward time $\tau$, respectively.

\begin{lemma}[Derivative of the $\sigma$-function] \label{lem:sigma_dot}
    Consider the functions \eqref{eq:sigma_function} together with optimal solutions $\xi_\tau(\cdot;\xi_0,u_e^*(\cdot),u_p^*(\cdot))$, $p_\tau(\cdot)$ satisfying the 
    dynamics \eqref{eq:dynamics} and the dynamics of the adjoint equation \eqref{eq:dot_p_equation_tau}, respectively.
If $\xi_\tau(\cdot;\xi_0,u_e^*(\cdot),u_p^*(\cdot))$ and $p_\tau(\cdot)$ are continuously differentiable, then it holds that
\begin{align}
    \dot{\sigma}(t) = -p_{x}(t) v_e \qquad \text{ and } \qquad \mathring{\sigma}_\tau(\tau) = p_{x_\tau}(\tau) v_e. \label{eq:sigma_dot}
\end{align}
\end{lemma}

\begin{proof}
We focus on $\sigma_\tau(\cdot)$ in the proof. The result for $\sigma(\cdot)$ follows from the relation $\frac{\mathrm{d} t}{\mathrm{d} \tau}=-1$.
For simplicity of exposition, we omit the time argument $\tau$ in the following derivations. With equations \eqref{eq:dynamics} and \eqref{eq:dot_p_equation_tau}, it holds that
\begin{align*}
    \mathring{\sigma} 
    &=\mathring{p}_{x_\tau} y_\tau + p_{x_\tau} \mathring{y}_\tau - \mathring{p}_{y_\tau} x_\tau - p_{y_\tau} \mathring{x}_\tau \\
    &=(\omega_e u_{e_\tau}^* p_{y_\tau}) y_\tau + p_{x_\tau} (-(\omega_e x_\tau u_{e_\tau}^* - v_e + v_p \cos u_{p_\tau})) - (-\omega_e u_{e_\tau}^* p_{x_\tau}) x_\tau - p_{y_\tau} (-(-\omega_e y_\tau u_{e_\tau}^* + v_p \sin u_{p_\tau}^*)) \\
    &=\omega_e u_{e_\tau}^* p_{y_\tau} y_\tau - p_{x_\tau} \omega_e x_\tau u_{e_\tau}^* + p_{x_\tau} v_e - p_{x_\tau} v_p \cos u_{p_\tau}^* + \omega_e u_{e_\tau}^* p_{x_\tau} x_\tau - p_{y_\tau} \omega_e y_\tau u_{e_\tau}^* + p_{y_\tau} v_p \sin u_{p_\tau}^* \\
    &= p_{x_\tau} v_e - p_{x_\tau} v_p \cos u_{p_\tau}^* + p_{y_\tau} v_p \sin u_{p_\tau}^*.
\end{align*}
Then, the assertion follows from \eqref{eq:opt_u_star}.
\end{proof}

Having established candidate optimal strategies in Lemma \ref{lem:opt_control_regular_E}, we next examine how they lead to candidate optimal solutions to the dynamics \eqref{eq:dynamics} with different initial (or terminal) conditions.

\subsection{Candidate Optimal Solutions Terminating on Usable Part with $x\neq 0$} \label{sec:solutions_ending_on_y_axis}

In this section we focus on solutions terminating on
$\{\xi \in \text{UP}| x\neq 0\}$ (where UP was defined in \eqref{eq:UP1}). To this end, we derive terminal conditions $\bp(T)$ of the adjoint variable corresponding to the adjoint equation \eqref{eq:adjoint_equation}, or equivalently, initial conditions $\bp_\tau(0)$ for \eqref{eq:dot_p_equation_tau} with respect to the time argument $\tau=T-t$. To make this point more precise, we derive trajectories $\xi_\tau(\cdot)$ starting on the usable part $\text{UP}$ and pointing to the interior of the game set $\mathcal{S}$ in backward time $\tau$. In this regard, here, initial conditions are with respect to $\tau=0$.

According to \cite[Corollary 5.2.3]{Lewin2012},
for $\xi \in \text{UP}$ where optimal trajectories terminate it holds that
\begin{align}
    \bp^{\top} {\bf t}(\xi) = \nabla G(\xi)^\top {\bf t}(\xi) \label{eq:terminate_condition}
\end{align}
and where ${\bf t}(\xi)\in \R^2$ denotes a tangent vector to the target set $\mathcal{C}$ at  $\xi$.
A tangent vector satisfying $\xi^\top {\bf t}(\xi)=0$ can be defined as
\begin{align}
    {\bf t}(\xi) = \left[ \begin{array}{r}
         y  \\
         -x 
    \end{array}\right].
\end{align}
Additionally, recall that the usable part $\text{UP}$ is defined in \eqref{eq:UP1} and the terminal costs $G(\xi)=0$ are defined in 
\eqref{eq:terminate_condition}.
Since $\nabla G(\xi)=0$, \eqref{eq:terminate_condition} leads to the condition
\begin{align}
    \bp^{\top} {\bf t}(\xi) = \left[ 
    \begin{array}{c}
       p_x  \\
       p_y   
    \end{array}
    \right]^\top 
    \left[ 
    \begin{array}{c}
       y  \\
       -x   
    \end{array}
    \right] = p_x y - p_y x = 0. \label{eq:px_py_relation2}
\end{align}
Since $\xi$ satisfies $|\xi| = \rho$ when the game terminates, we can represent 
$\xi \in \text{UP}$ through through $\theta_\xi \in (-\pi,\pi]$ using the transformation \eqref{eq:sin_cos_transform}:
\begin{align}
    x= \rho \sin(\theta_\xi), \qquad y = \rho \cos(\theta_\xi). \label{eq:terminal_condition}
\end{align}
In this representation, 
\eqref{eq:px_py_relation2}
can be equivalently stated as
\begin{align}
    p_x \cos(\theta_\xi) - p_y \sin(\theta_\xi) = 0
    \qquad \text{or} \qquad
    p_x = p_y \tan(\theta_\xi). \label{eq:p_x_p_y_relation}
\end{align}

To derive a second condition, we use the Hamiltonian \eqref{eq:Ham} and the condition \eqref{eq:H_0_condition}.
Accordingly, for regular parts on the optimal trajectory, for points on UP, it holds that
\begin{align*}
0=  &p_x(-\omega_e y u_e^* + v_p \sin u_p^*) + p_y(\omega_e x u_e - v_e + v_p \cos u_p^*) +1\\
=   &p_x  v_p \sin u_p^*  + p_y( - v_e + v_p \cos u_p^*) +1
\end{align*}
and where the second equality follows from \eqref{eq:px_py_relation2}.
Note that we have not specified $u_e^*$ and $u_p^*$  yet.
With \eqref{eq:p_x_p_y_relation}, the last expression can be rewritten as
\begin{align*}
    0=&p_y \tfrac{\sin \theta_\xi}{\cos \theta_\xi} v_p \sin (u_p^*)  + p_y( - v_e + v_p \cos (u_p)) +1 
\end{align*}
and thus
\begin{align}
  p_y  =&- \frac{1}{v_p \left[\tfrac{\sin(\theta_\xi)}{\cos(\theta_\xi)}\sin(u_p^*) + \cos(u_p^*)\right]-v_e} \nonumber \\
  =&- \frac{\cos(\theta_\xi)}{v_p [\sin(\theta_\xi)\sin(u_p^*) + \cos(\theta_\xi) \cos(u_p^*)]- \cos(\theta_\xi) v_e} \nonumber \\
  =&- \frac{\cos(\theta_\xi)}{v_p \cos(\theta_\xi -u_p^*)- \cos(\theta_\xi) v_e}. \label{eq:terminal_cond_py}
\end{align}
Again using \eqref{eq:p_x_p_y_relation}, $p_x$ satisfies
\begin{align}
\begin{split}
    p_x 
    &= - \tfrac{\sin \theta_\xi}{\cos \theta_\xi}  \frac{\cos(\theta_\xi)}{v_p \cos(\theta_\xi -u_p^*)- \cos(\theta_\xi) v_e} \\
    &= - \frac{\sin(\theta_\xi)}{v_p \cos(\theta_\xi -u_p^*)- \cos(\theta_\xi) v_e}.
\end{split} \label{eq:terminal_cond_px}
\end{align}

Under the assumption that $u_{e_\tau}^*(\cdot)$ is constant (which is the case if  $\sigma_{\tau}(\cdot)=p_{x_\tau}(\cdot) y_{\tau}(\cdot)-p_{y_\tau}(\cdot)  x(\cdot)$ does not change sign according to \eqref{eq:switch_function} and \eqref{eq:sigma_function}), the solution of \eqref{eq:dot_p_equation_tau} with initial condition \eqref{eq:terminal_cond_py}, \eqref{eq:terminal_cond_px} is given by 
\begin{align}
\begin{split}
    \bp_{\tau}(\tau) &= \left[ \begin{array}{rr}
        \cos(-\omega_e u_{e_\tau}^*(\tau) \tau) & -\sin(-\omega_e u_{e_\tau}^*(\tau) \tau) \\
        \sin(-\omega_e u_{e_\tau}^*(\tau)\tau)  & \cos(-\omega_e u_{e_\tau}^*(\tau)\tau) 
    \end{array}\right]
    \left[\begin{array}{c}
         \tfrac{\sin \theta_\xi}{v_p \cos(\theta_\xi -u_{p_\tau}^*(0))- \cos(\theta_\xi) v_e}  \\
         \tfrac{\cos \theta_\xi}{v_p \cos(\theta_\xi -u_{p_\tau}^*(0))- \cos(\theta_\xi) v_e}  
    \end{array} \right] \\
    &=\left[\begin{array}{c}
         -\tfrac{\sin (-\omega_e u_{e_\tau}^*(\tau) \tau - \theta_\xi)}{v_p \cos(\theta_\xi -u_{p_\tau}^*(0))- \cos(\theta_\xi) v_e}  \\
         \tfrac{\cos (-\omega_e u_e^* \tau - \theta_\xi)}{v_p \cos(\theta_\xi -u_p^*(0))- \cos(\theta_\xi) v_e}  
    \end{array} \right] 
=    \left[\begin{array}{c}
         \tfrac{\sin (\theta_\xi + \omega_e u_{e_\tau}^*(\tau) \tau)}{v_p \cos(\theta_\xi -u_{p_\tau}^*(0))- \cos(\theta_\xi) v_e}  \\
         \tfrac{\cos (\theta_\xi + \omega_e u_{e_\tau}^*(\tau) \tau)}{v_p \cos(\theta_\xi -u_{p_\tau}^*(0))- \cos(\theta_\xi) v_e}  
    \end{array} \right]=\left[  \begin{array}{r}
        p_{x_\tau}  \\
        p_{y_\tau}   
    \end{array} \right].
    \end{split} \label{eq:adjoint_variables}
\end{align}

Using this information in \eqref{eq:opt_u_star}, i.e., 
\begin{align}
\frac{1}{v_p \cos(\theta_\xi -u_{p_\tau}^*(0))- \cos(\theta_\xi) v_e}
\left[\begin{array}{c}
         \sin (\theta_\xi+\omega_e u_{e_\tau}^*(\tau) \tau )  \\
         \cos (\theta_\xi+\omega_e u_{e_\tau}^*(\tau) \tau )
    \end{array} \right]
    = c \left[ \begin{array}{c}
     \sin(u_{p_\tau}^*(\tau))  \\
     \cos(u_{p_\tau}^*(\tau)) 
\end{array} \right]  
= -c \left[ \begin{array}{c}
     \sin(u_{p_\tau}^*(\tau)-\pi)  \\
     \cos(u_{p_\tau}^*(\tau)-\pi) 
\end{array} \right] \label{eq:u_p_cases}
\end{align}
implies that either
$u_{p_\tau}^*(\tau)=  \theta_\xi+ \omega_e u_{e_\tau}^*(\tau) \tau$
or $u_{p_\tau}^*(\tau)= \pi + \theta_\xi+ \omega_e u_{e_\tau}^*(\tau) \tau$.
From Remark \ref{rem:u_p_star} we know that $u_{p_\tau}^*(0) = \pi+\theta_\xi$ ensures that the pursuer points to the evader, i.e., $f(\xi_{\tau}(0),u_{e_\tau}^*(0),\pi+\theta_\xi)^\top \xi_{\tau}(0)<0$
(while $f(\xi_{\tau}(0),u_{e_\tau}^*(0),\theta_\xi)^\top \xi_{\tau}(0)>0$),
and thus we can conclude the optimal pursuer strategy
\begin{align}
    u_{p_\tau}^*(\tau)= \pi + \theta_\xi+ \omega_e u_{e_\tau}^*(\tau) \tau. \label{eq:up_star1}
\end{align}

As a next step we turn our attention to $u_{e_\tau}^*(\cdot)$.
According to \eqref{eq:px_py_relation2}, on the target set $\mathcal{C}$, it holds that $\sigma_\tau(0)=p_{x_\tau}(0) y_{\tau}(0) - p_{y_\tau}(0) x_{\tau}(0) = 0$. Indeed, we have already observed in 
Remark \ref{rem:u_e_independence} that on $\mathcal{C}$ the game ends or continues independent of the evader's strategy. To derive $u_{e_\tau}^*(\cdot)$ for $\xi \in \mathcal{S}\backslash \mathcal{C}$, we use Lemma \ref{lem:sigma_dot} and the general solution \eqref{eq:adjoint_variables} of the adjoint variables. 
From the mean value theorem, it follows that for $\tau_1>0$ sufficiently small, there exists $\tau_2 \in (0,\tau_1)$ such that
\begin{align}
    \mathring{\sigma}_\tau(\tau_2) = \frac{\sigma_\tau(\tau_1)-\sigma_\tau(0)}{\tau_1-0},
\end{align}
i.e., with $\sigma(0)=0$ and $p_{x_\tau}(\cdot)$ defined in \eqref{eq:adjoint_variables} it holds that 
\begin{align}
    \sigma_\tau(\tau_1) = \tau_1 \mathring{\sigma}_\tau(\tau_2) =  \tau_1 p_{x_\tau}(\tau_2) v_e
    = \tau_1 v_e \tfrac{\sin (\theta_\xi + \omega_e u_e^*(\tau_2) \tau_2)}{v_p \cos(\theta_\xi -u_p^*(0))- \cos(\theta_\xi) v_e}.
\end{align}
The denominator satisfies
\begin{align}
    v_p \cos(\theta_\xi -\pi-\theta_\xi)- \cos(\theta_\xi) v_e = -v_p - \cos(\theta_\xi) v_e > -v_p + v_e >0
\end{align}
and where we have used that $\cos(\theta_\xi) \leq -\frac{v_p}{v_e}$ according to the definition of the usable part.
Since $\cos(\theta_\xi) \leq -\frac{v_p}{v_e}<0$, it additionally holds that $\theta_\xi \in [-\pi,-\frac{\pi}{2}) \cup  (\frac{\pi}{2},\pi]$. 
Thus, $x_0=\rho \sin(\theta_\xi)>0$ for $\theta_\xi \in (\frac{\pi}{2},\pi)$, $x_0=\rho\sin(\theta_\xi)<0$ for $\theta_\xi \in (-\pi,-\frac{\pi}{2})$ and $\theta_\xi=\pi$ corresponds to the $y$-axis which we have excluded from the discussion in this section. 
For sufficiently small $\tau$, we can conclude that $\sigma(\tau)>0$ for $x_0>0$ and $\sigma(\tau)<0$ for $x_0<0$. Accordingly, for solutions starting on the usable part $\xi_{0}\in \{\xi \in \text{US}|x\neq 0\}$ in backward time, the optimal strategy of the evader satisfies
\begin{align}
    u_{e_\tau}^*(\tau) = -\sign(x_0). \label{eq:opt_u_e_star}
\end{align}
Note that for $x_0\neq 0$, $\sign(x_0)$ is uniquely defined. The definition of $u_{e_\tau}(\cdot)$ also takes into account the fact that on $\text{US}$ the value of $u_e^*$ does not matter since it is multiplied with $\tau=0$ according to \eqref{eq:adjoint_variables}.
We can additionally observe that on the boundary of the game set $y=\rho \cos(\theta_\xi)$ and thus according to \eqref{eq:UP1} it holds that $\cos(\theta_\xi) < -\frac{v_p}{v_e}$. Hence, the denominator in \eqref{eq:adjoint_variables} is unequal to zero and $\bp(\cdot)$ is well defined.

We summarize the derivations in this section in the following theorem.

\begin{theorem} \label{thm:solutions_usable_part}
    Consider Problem \ref{prob:game_of_degree} and let $\xi_0\in \{\xi \in \text{US}| x\neq 0 \}$ denote an initial condition in backward time (i.e., at $\tau = 0$). Moreover, let $\theta_{\xi_0} \in (-\pi,\pi]$ be defined such that
    \begin{align}
        \xi_0 = \left[ \begin{array}{c}
             \rho \sin(\theta_{\xi_0})  \\
             \rho \sin(\theta_{\xi_0})
        \end{array}\right].
    \end{align}
    Then the optimal strategy of the evader satisfies $u_{e_\tau}^*(\tau)=-\sign(x_0)$ and the optimal strategy of the pursuer satisfies
    \begin{align}
        u_p^*(\tau)= \pi + \theta_{\xi_0} -\sign(x_0) \omega_e \tau. 
    \end{align}
    Moreover, the optimal solutions $\xi_\tau(\cdot)$ and $\bp_\tau(\cdot)$ satisfy
    \begin{align}
    \xi_\tau(\tau) &= \frac{1}{\omega_e} \left[ \begin{array}{c}
     -\sign(x_0) v_e  +\sign(x_0)v_e  \cos \left(\tau \omega_e  \right)+\rho \omega_e \sin \left(\theta_{\xi_0} -\sign(x_0)\tau \omega_e  \right)+\tau v_p  \omega_e \sin \left(\theta_{\xi_0} -\sign(x_0)\tau \omega_e  \right) \\
      v_e  \sin \left(\tau \omega_e \right)+\rho  \omega_e  \cos \left(\theta_{\xi_0} -\sign(x_0)\tau  \omega_e  \right)+\tau v_p  \omega_e \cos \left(\theta_{\xi_0} -\sign(x_0)\tau \omega_e  \right)
    \end{array} \right], \label{eq:solutions_usable_part} \\
    \bp_{\tau}(\tau) & =  \frac{1}{-v_p - \cos(\theta_{\xi_0}) v_e}    \left[\begin{array}{c}
         \sin (\theta_{\xi_0} -\sign(x_0) \omega_e  \tau)  \\
         \cos (\theta_{\xi_0} -\sign(x_0) \omega_e  \tau)
    \end{array} \right],  \label{eq:solution_adjoint_variables_p1}
\end{align}
respectively, for all $\tau \in \R_{\geq 0}$ such that $\xi_{\tau}(\tau) \in \mathcal{S}$. 
\hfill $\lrcorner$
\end{theorem}

\begin{proof}
The proof follows immediately from the derivations in this section.
The optimal strategies for the evader and pursuer have been derived in \eqref{eq:opt_u_e_star}
and \eqref{eq:up_star1}, respectively.
With $u_{e_\tau}^*(\cdot)$ and $u_{p_\tau}^*(\cdot)$ defined, $\bp_\tau(\cdot)$ follows from \eqref{eq:adjoint_variables}. 
Using the evader and pursuer strategies in the dynamics \eqref{eq:dynamics} in backward time leads to the ordinary differential equation
\begin{align}
    \mathring{\xi}(\tau) = -f(\xi(\tau),u_e^*,u_p^*(\tau)) = \left[ 
    \begin{array}{c}
        -\sign(x_0) \omega_e y  - v_p \sin (\pi+\theta_{\xi_{0}} -\sign(x_0) \omega_e  \tau )\\
        \sign(x_0)\omega_e x  + v_e - v_p \cos (\pi+\theta_{\xi_0} -\sign(x_0) \omega_e  \tau )
    \end{array}
    \right],
\end{align}
which can be solved analytically,\footnote{The differential equation has been solved in MATLAB.} leading to the expression \eqref{eq:solutions_usable_part} and completing the proof.
\end{proof}

Solutions characterized in Theorem \ref{thm:solutions_usable_part} and emanating from the usable part in backward time for the parameters $v_p=1$, $v_e=2$, $\rho=1$ and $\omega_e=2$ are shown in Figure \ref{fig:solution_usable_part}.
In particular, Theorem \ref{thm:solutions_usable_part} corresponds to part (A) in
Figure \ref{fig:solution_game_degree_intro_new}. We next focus on part (B) in Figure \ref{fig:solution_game_degree_intro_new}. 
\begin{figure}[htb!]
    \centering
    \includegraphics[width = 0.4\columnwidth]{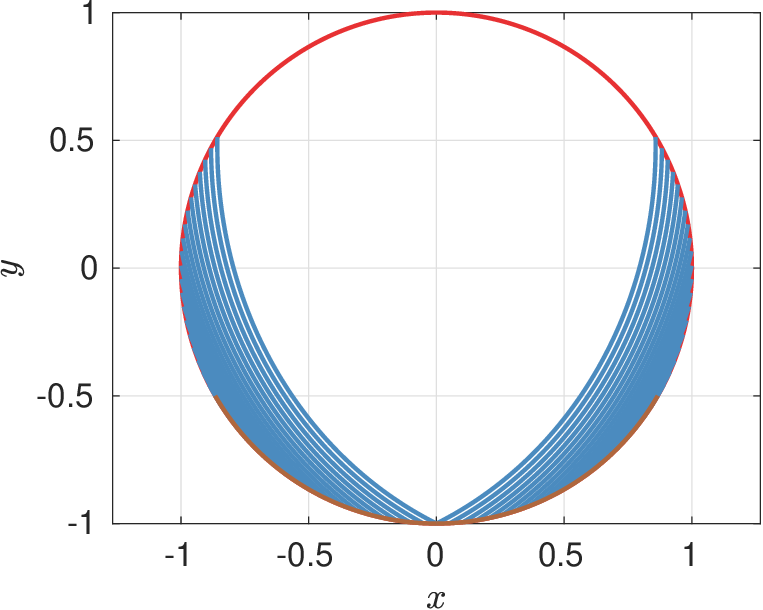}
    \caption{Solutions of the game of degree emanating from the usable part in backward time characterized through Theorem \ref{thm:solutions_usable_part} for the parameters $v_p=1$, $v_e=2$, $\rho=1$ and $\omega_e=2$.}
    \label{fig:solution_usable_part}
\end{figure}

\subsection{Trajectories Converging to $y$-axis in Finite Time}

We now derive solutions highlighted through (B) in Figure \ref{fig:solution_game_degree_intro_new}.
In Lemma \ref{lem:strategy_on_y_axis} we have established solutions for games starting on the set $\xi_0 \in \{ \xi \in \mathcal{S}|x =0, y<0\}$.
In particular, from \eqref{lem:strategy_on_y_axis}
we know that 
\begin{align}
    V(\xi_0) = \frac{\rho+y_0}{v_e-v_p}, \qquad \xi_0 \in \{ \xi \in \mathcal{S}|x =0, y<0\}.
\end{align}
We use this fact to define a new auxiliary game of degree, ending on the $y$-axis and considering the function $G:\mathcal{S}\rightarrow\R$,
\begin{align}
    G_y(\xi) = \frac{\rho+y}{v_e-v_p} \label{eq:terminal_cost_y}
\end{align}
as a terminal cost. 

\begin{problem} \label{prob:game_of_degree2}
Consider the dynamics \eqref{eq:dynamics} with state $\xi\in \R^2$, inputs $u_e\in [-1,1]$ and $u_p\in (-\pi,\pi]$, and defined through parameters $\omega_e,v_e,v_p\in \R_{>0}$ with $v_p<v_e$. Additionally, consider running costs $L(\xi,u_e,u_p)= 1$ and terminal costs defined in \eqref{eq:terminal_cost_y}.
Then we define the (auxiliary) game of degree
\begin{align}
\begin{split}
V_y(\xi_0) = \min_{u_e} \max_{u_p} & \int_0^T L(\xi,u_e,u_p) \, \mathrm{d}t + G_y(\xi(T)) \\
\text{subject to } \quad & \dot{\xi}(t) = f(\xi(t), u_e(t), u_p(t)), \quad \xi(0) = \xi_0, \\
& u_e(t) \in [-1,1], \quad t \in [0,T],\\
& u_p(t) \in \mathbb{R}, \quad t \in [0,T],\\
& \xi(0) \in \R^2, \quad \xi(T)^\top \left[\begin{smallmatrix}
    1 \\
    0
\end{smallmatrix}\right] = 0,
\end{split} \label{eq:problem_constraints2}
\end{align}
ending on the $y$-axis.
\hfill $\triangleleft$
\end{problem}

Note that the auxiliary game of degree \eqref{eq:problem_constraints2} admits the same Hamiltonian \eqref{eq:Ham} and the same adjoint equation \eqref{eq:adjoint_equation} as the game of degree \eqref{eq:problem_constraints}.
Solutions of the auxiliary game of degree \eqref{eq:problem_constraints2} ending on the $y$-axis can be concatenated with the solution \eqref{lem:strategy_on_y_axis} to recover solutions of \eqref{eq:problem_constraints}.

\begin{theorem} \label{thm:opt_strat_game2}
Consider the (auxiliary) game of degree defined in Problem \ref{prob:game_of_degree2}. Let $\xi_0\in \{\xi \in \mathcal{S}| x= 0, \ y<0 \}$ denote an initial condition in backward time. 
Then the optimal pursuer strategy is given by
\begin{align}
    u_{p_\tau}^*(\tau) &= u_{e_\tau}^*(\tau) \omega_e \tau. \label{eq:opt_u_p_prob2}
\end{align}
Moreover, for constant $u_{e_\tau}^* \in \{-1,1\}$, solutions of
\eqref{eq:dynamics} in backward time
with $\xi_0 \in \{\xi \in \mathcal{S}| x=0, \, y<0\}$ are given by
\begin{align}
    \xi_{\tau}(\tau;\xi_0,u_{e_\tau}^*,u_{p_\tau}^*(\tau))
    &=
    \begin{bmatrix}
        u_{e_\tau}^* (y_0 \sin(\tau \omega_e) + \frac{v_e}{\omega_e} - \tau v_p \sin(\tau \omega_e) -\frac{v_e}{\omega_e}\cos(\tau \omega_e)) \\
        y_0 \cos(\tau \omega_e) + \frac{v_e}{\omega_e} \sin(\tau \omega_e) - \tau v_p \cos(\tau \omega_e)
    \end{bmatrix} \label{eq:sol_xP2}
\end{align}
for all $\tau\in \R_{\geq 0}$ such that $u_{e_\tau}^*\in \sign(-x_{\tau}(\tau))$.
\hfill $\lrcorner$
\end{theorem}

Before we prove this result, note that for 
\begin{align}
    \xi_0=\lim_{\theta_{\xi_0}\rightarrow \pi} \left[ \begin{array}{c}
         \rho \sin(\theta_{\xi_0})  \\
          \rho \cos(\theta_{\xi_0})
    \end{array} \right] \label{eq:continuous_transition}
\end{align}
the strategies established in Theorem \ref{thm:solutions_usable_part} are consistent with the strategies in Theorem \ref{thm:opt_strat_game2}, i.e., in the limit, the strategies in 
Theorem \ref{thm:opt_strat_game2} 
and Theorem \ref{thm:solutions_usable_part} coincide.

\begin{proof}
As in the previous section, we use conditions 
\eqref{eq:terminate_condition} and
\eqref{eq:H_0_condition}
to derive a terminal condition for the adjoint variables.
With
\begin{align}
    \nabla G_y(\xi) = \left[ \begin{array}{c}
         0  \\
         \frac{1}{v_e-v_p} 
    \end{array}\right] \qquad \text{and} \qquad {\bf t}(\xi) = \left[ \begin{array}{c}
         0  \\
         1 
    \end{array}\right],
\end{align}
condition \eqref{eq:terminate_condition} implies that
\begin{align*}
    p_y=\bp^{\top} {\bf t}(\xi) = \nabla G_y(\xi)^\top {\bf t}(\xi) = 
         \frac{1}{v_e-v_p}.
\end{align*}

As a next step, we consider an \emph{ansatz} for $p_x$ where $p_x=0$. For $p_x=0$, the adjoint variables with constant $u_{e_\tau}^*\in \{-1,1\}$ satisfy 
\begin{align}
\begin{split}
    \bp_{\tau}(\tau) &= \left[ \begin{array}{rr}
        \cos(-\omega_e u_{e_\tau}^* \tau) & -\sin(-\omega_e u_{e_\tau}^* \tau) \\
        \sin(-\omega_e u_{e_\tau}^*\tau)  & \cos(-\omega_e u_{e_\tau}^*\tau) 
    \end{array}\right]
    \left[\begin{array}{c}
         0  \\
         \tfrac{1}{v_e-v_p}  
    \end{array} \right] \\
    &= \frac{1}{v_e-v_p} \left[\begin{array}{r}
         -\sin (- \omega_e u_{e_\tau}^* \tau)  \\
         \cos (- \omega_e u_{e_\tau}^* \tau)  
    \end{array} \right] 
    = \frac{1}{v_e-v_p} \left[\begin{array}{r}
         \sin (\omega_e u_{e_\tau}^* \tau)  \\
         \cos (\omega_e u_{e_\tau}^* \tau)  
    \end{array} \right]. 
\end{split} \label{eq:adjoint_solution_P2}
\end{align}
Using the same arguments as in
\eqref{eq:u_p_cases}, for $p_x=0$,
we conclude that $u_{p_\tau}^*(\tau) = \omega_e u_{e_\tau}^* \tau$ or $u_{p_\tau}^*(\tau) = \pi + \omega_e u_{e_\tau}^* \tau$.

With the derivations so far, condition \eqref{eq:H_0_condition} on the Hamiltonian implies that
\begin{align*}
    0&=\bp_{\tau}(0)^\top f(\xi_{\tau}(0),u_{e_\tau}^*,u_{p_\tau}(0)) +L(\xi(0),u_{e_\tau}^*,u_{p_\tau}^*(0)) \\
    &= p_{x_\tau}(0) (-\omega_e y_{\tau}(0) u_{e_\tau}^* + v_p \sin (u_{p_\tau}^*(0))) + p_{y_\tau}(0) (\omega_e x(0) u_{e_\tau}^* - v_e + v_p \cos (u_{p_\tau}^*(0))) +1 \\
    &=  p_{y_\tau}(0) ( - v_e + v_p \cos (u_{p_\tau}^*(0))) +1 \\
    &= \frac{1}{v_e-v_p} (- v_e + v_p \cos(u_{p_\tau}^*(0) ) +1.
\end{align*}
Hence, 
\begin{align}
    u_p^*(\tau) = \omega_e u_{e_\tau}^* \tau,
\end{align}
which confirms that the \emph{ansatz} $p_x=0$ was correct.

As a next step we can focus on the dynamics
\eqref{eq:dynamics} in backward time
with $u_{e_\tau}^*\in\{-1,1\}$, $u_p^*(\tau) =\omega_e u_e^* \tau$ and terminal condition $\xi_0 \in \{\xi \in \mathcal{S}| x=0, \, y<0\}$, i.e., we can focus on the differential equation 
\begin{align}
    \mathring{\xi}
    &=
    \begin{bmatrix}
        \omega_e y u_{e_\tau}^* - v_p \sin (\omega_e u_{e_\tau}^* \tau)\\
        -\omega_e x u_{e_\tau}^* + v_e - v_p \cos (\omega_e u_{e_\tau}^* \tau)
    \end{bmatrix}, \qquad \xi_0 \in \left[\begin{array}{c}
         0 \\
         y_0 
    \end{array}\right].  \label{eq:differential_equation_solP2}
\end{align}
Here, the $x$-component of the differential equation \eqref{eq:differential_equation_solP2} satisfies $\mathring{x}(0)<0$ for $u_{e_\tau}(\tau)=1$ and $\mathring{x}(0)>0$ for $u_{e_\tau}=-1$ for all $\xi_0 \in \{\xi \in \mathcal{S}| x=0,\, y<0\}$. Accordingly, $u_e^*=-1$ is optimal for $x>0$ while $u_e^*=1$ is optimal for $x<0$.

The initial value problem \eqref{eq:differential_equation_solP2} can be solved analytically\footnote{Here, the solution is obtained through MATLAB.} and the solution is given by \eqref{eq:sol_xP2}, which completes the proof.
\end{proof}

\begin{remark}
    By construction, for all $\xi_0 \in \{\xi \in \mathcal{S}|x=0,\ y<0\}$, the pair $(\xi_\tau(\cdot),\bp_\tau(\cdot))$ defined in \eqref{eq:adjoint_solution_P2} and \eqref{eq:sol_xP2}, respectively, satisfy \eqref{eq:Ham_cond1} for all $\tau\in \R_{\geq 0}$ such that $\xi(\tau)\in \{\xi\in\mathcal{S}| x u_e^* \leq 0\}$.
    \hfill $\diamond$
\end{remark}

Solutions characterized through Theorem \ref{thm:opt_strat_game2} for different parameter selection are shown in Figure \ref{fig:solution_new_game}.
For Figure \ref{fig:solution_new_game}, right, we observe that (in combination with Theorem \ref{thm:solutions_usable_part}) solutions fill the game set while in Figure \ref{fig:solution_new_game} left, there is still a gap on the $y$-axis. This gap, denoted by (C) in Figure \ref{fig:solution_game_degree_intro_new}, arises when the speed ratio satisfies $\frac{v_p}{v_e} > \frac{1}{2}$, and will be addressed in the next section.
\begin{figure}[htb!]
    \centering
    \includegraphics[width = 0.45\columnwidth]{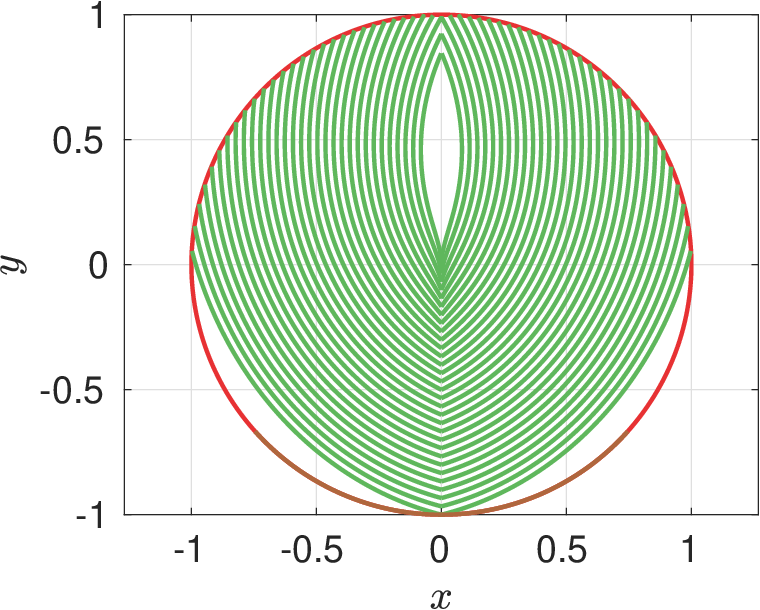} 
    \hspace{0.4cm}
    \includegraphics[width = 0.45\columnwidth]{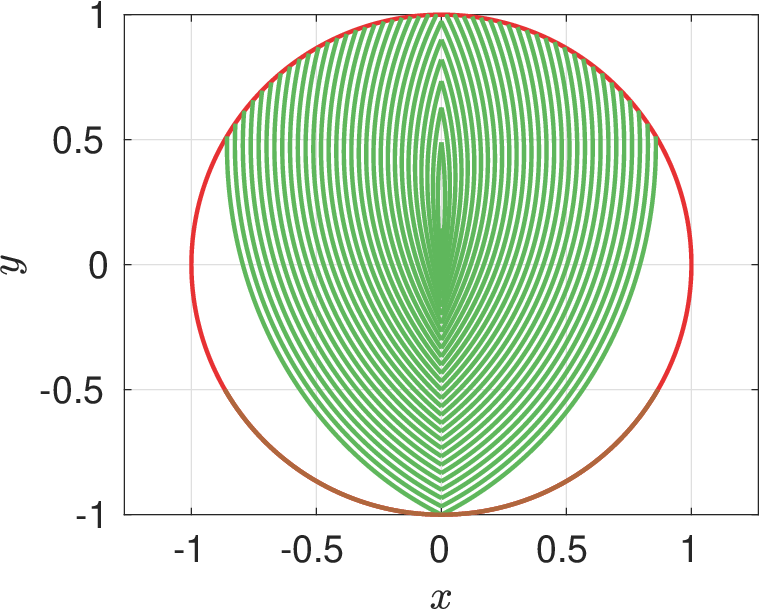}
    \caption{Trajectories characterized through Theorem \ref{thm:opt_strat_game2} using the parameter selection $\omega_e=2$, $v_p=1$, $\rho=1$ and $v_e=1.5$ (left) and $v_e=2$ (right), respectively.}
    \label{fig:solution_new_game}
\end{figure}
From Figure \ref{fig:solution_new_game} and Theorem \ref{thm:opt_strat_game2}, we additionally observe that $\{\xi \in \mathcal{S}|x=0,\ y<0\}$ is a universal surface \cite[Sec. 9.5.2]{Lewin2012}. We also observe that the set $\{\xi \in \mathcal{S}|x=0,\ y>0\}$ or part of the set is a dispersal surface \cite[Sec. 9.5.2]{Lewin2012}.

\begin{remark}
Note that we have not discussed solutions starting in $\xi_0=0$ yet. Through continuity arguments, it can be concluded that the solution can be characterized as in Lemma \ref{lem:strategy_on_y_axis}.
\hfill $\diamond$
\end{remark}

\subsection{Extension of Solution \eqref{eq:sol_xP2} for $y\geq 0$} \label{sec:weird_part}

In the preceding section we have argued how solutions starting on the set $\{\xi \in \mathcal{S}|x=0,\ y\leq 0\}$ in backward time can be constructed. In this section, we expand the discussion to $\xi_0 \in \{\xi \in \mathcal{S}|x=0,\ y> 0\}$ and derive corresponding solutions in backward time.
We extend solution \eqref{eq:sol_xP2} derived in Theorem \ref{thm:opt_strat_game2} based on the pursuer's strategy \eqref{eq:opt_u_p_prob2}.
Here, we focus on the case $u_{e_{\tau}}^*(\tau)=-1$, noting that, due to symmetry (see Corollary \ref{cor:symmetry}), the case $u_{e_{\tau}}^*(\tau)=1$ can be discussed in the same way.

As a first step, note that for $\xi_0 \in\{\xi \in \mathcal{S}|x=0\}$, $u_{e_\tau}^*=-1$ and $u_{p_\tau}^*(0)=0$, the right-hand side of the dynamical system \eqref{eq:dynamics} in backward time reduces to
\begin{align}
 \mathring{\xi}_{\tau} =   -f(\xi_0, -1, 0)
    &=
    \begin{bmatrix}
        -\omega_e y_0  \\
         v_e - v_p
    \end{bmatrix}. \label{eq:direction_y0}
\end{align}
Depending on the sign of the initial value $y_0$, we observe that the corresponding solution satisfies
\begin{align}
\left\{
\begin{array}{ll}
    \mathring{x}(0;\xi_0,-1,0) <0,& \qquad \text{for } y_0 >0,\\
    \mathring{x}(0;\xi_0,-1,0) =0,& \qquad \text{for } y_0 =0,\\
    \mathring{x}(0;\xi_0,-1,0) >0,& \qquad \text{for } y_0 <0
\end{array} \right. \label{eq:sign_based_on_y0}
\end{align}
(and where we have used the derivative for illustration purposes instead of the more general argument in terms of the direction  $-\omega_e y_0$ in \eqref{eq:direction_y0}).
Due to this fact, we have only discussed solutions corresponding to $y_0<0$ in Theorem \ref{thm:opt_strat_game2}. However, \eqref{eq:sol_xP2} is also well defined for $y_0\geq 0$ and solutions satisfy condition \eqref{eq:H_0_condition} for all $\tau \in \R_{\geq 0}$ by construction. 
Solutions corresponding to the parameters $v_p=1$, $v_e=1.5$ and $\rho=1$ and the two values $\omega_e=1$ and $\omega_e=2$, respectively, are shown in  Figure \ref{fig:solution_new_game2w_1}. 
\begin{figure}[htb!]
    \centering
    \includegraphics[width = 0.23\columnwidth]{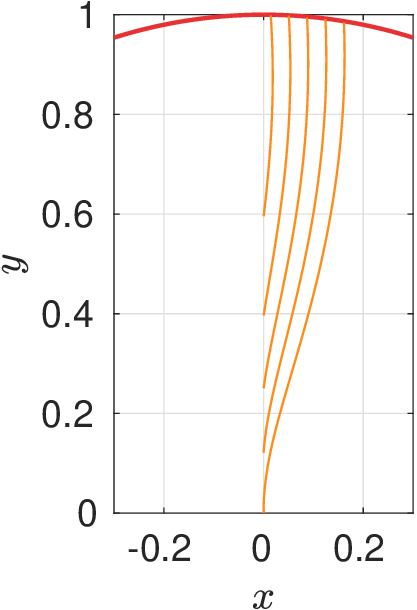}
    \includegraphics[width = 0.23\columnwidth]{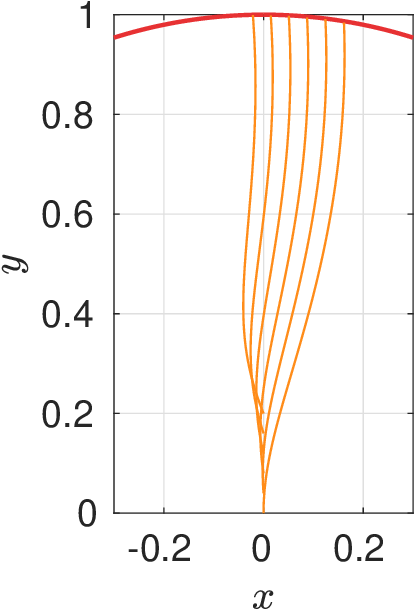}
\hfill
    \includegraphics[width = 0.23\columnwidth]{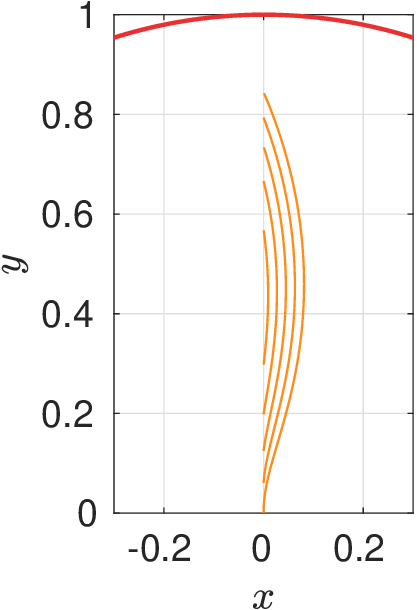}
    \includegraphics[width = 0.23\columnwidth]{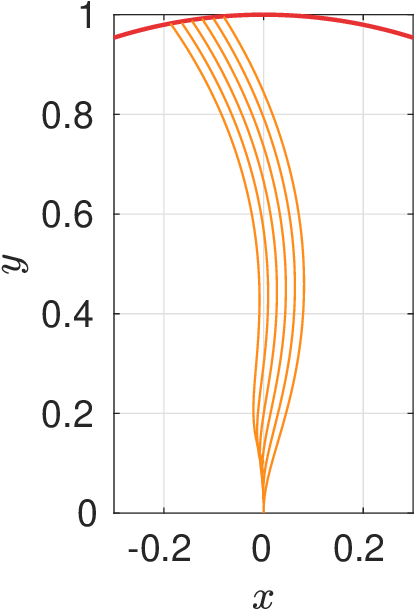}
    \caption{Trajectories characterized through Theorem \ref{thm:opt_strat_game2} starting on the $y$-axis with $y_0>0$. Here, we focus on $u_{e_\tau}^*(\tau)=-1$ and use the parameter selection $v_p=1$, $\rho=1$, $v_e=1.5$ and $\omega_e=1$ (left) and $\omega_e=2$ (right), respectively. Only the components of the solutions satisfying $x_\tau(\tau)\geq 0$ are relevant for the solution of the game of degree.}
    \label{fig:solution_new_game2w_1}
\end{figure}

As indicated in \eqref{eq:sign_based_on_y0}, for $y_0>0$, solutions initially enter the half-space $\{\xi \in \mathcal{S}|x<0 \}$. While some solutions remain in $\{\xi \in \mathcal{S}|x<0 \}$ indefinitely, we focus on parts that change the sign of the $x$-component. Recalling the explicit expression of the solution \eqref{eq:sol_xP2}, we investigate zeros of the equation 
\begin{align}
0= u_{e_\tau}^*(y_0 \sin(\tau \omega_e) + \frac{v_e}{\omega_e} - \tau v_p \sin(\tau \omega_e) -\frac{v_e}{\omega_e}\cos(\tau \omega_e)), 
\end{align}
or equivalently (since $u_{e_\tau}^*=-1$ by assumption), zeros of the expression
\begin{align}
0= \left(\frac{y_0}{v_e} \omega_e - \tau \omega_e \frac{v_p}{v_e}\right) \sin(\tau \omega_e) + 1  -\cos(\tau \omega_e). \label{eq:arbitrary_zeroes_equation}
\end{align}
To simplify \eqref{eq:arbitrary_zeroes_equation} we introduce the scaled parameters
\begin{align}
    \tilde{\tau} =\tau \omega_e, \qquad 
    \tilde{y}_0 =\frac{\omega_e}{v_e} y_0 , \qquad v_r = \tfrac{v_p}{v_e} \label{eq:scaled_params}
\end{align}
and thus
  $0=(\tilde{\tau}  v_r-\tilde{y}_0) \sin(\tilde{\tau}) -   1-\cos(\tilde{\tau})$. 
Note that this coordinate transformation reduces the number of parameters from five to three and \eqref{eq:scaled_params} corresponds to a positive linear scaling of the time argument $\tau$ and the initial condition $y_0$, and $v_r<1$ represents the speed ratio of the pursuer and the evader.
For a fixed speed ratio $v_r$, we define the function 
\begin{align}
    \Gamma_{v_r}(\tilde{\tau},\tilde{y}_0) =   (\tilde{\tau}  v_r-\tilde{y}_0) \sin(\tilde{\tau}) +   1-\cos(\tilde{\tau}). \label{eq:Gamma}
\end{align}
The number of zeros of the function $\Gamma_{v_r}:[0,2\pi) \times [0,\frac{\omega_e\rho}{v_e}] \rightarrow \R$ depends on the speed ratio $v_r$. In particular, while it follows directly that $\Gamma_{v_r}(0,0)=0$ for all $v_r \in (0,1)$, the ratio $v_r\leq\frac{1}{2}$ or $v_r> \frac{1}{2}$ decides if $\Gamma_{v_r}(\cdot,\cdot)$ has additional zeros. Figure \ref{fig:solution_new_game2w} shows contour lines of $\Gamma_{v_r}$ for different parameters $v_r$.
\begin{figure}[htb!]
    \centering
    \begin{overpic}[width = 0.32\columnwidth]{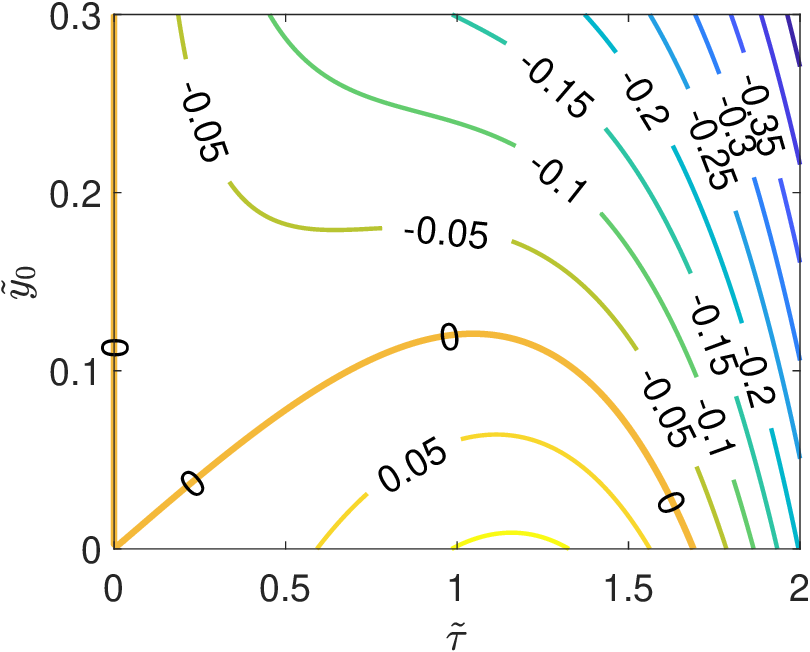}
    \put(57,37.5){$+$}
    \put(50,30.5){\footnotesize{$(\tilde{\tau}^\#,\tilde{y}_0^\#)$}}
    \end{overpic}
    \hfill
        \includegraphics[width = 0.32\columnwidth]{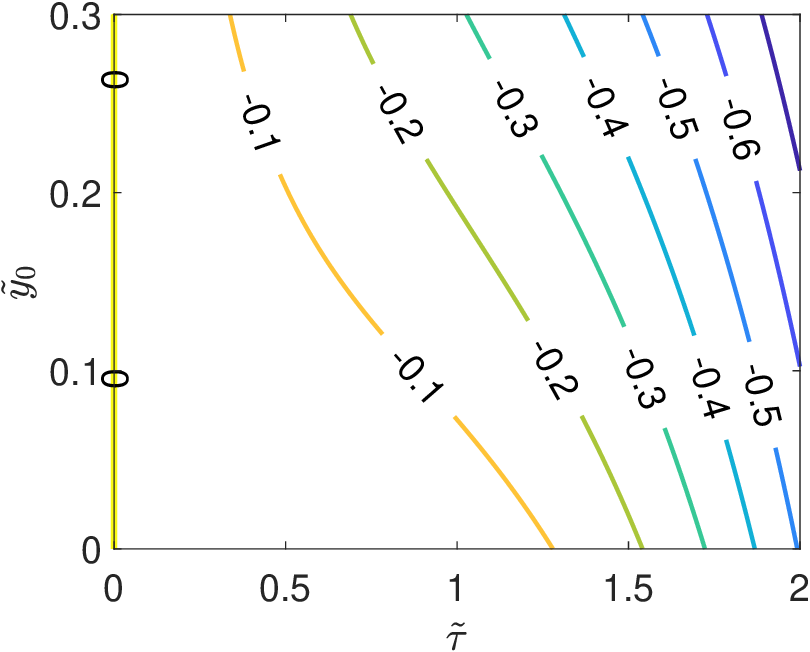}
        \hfill
        \includegraphics[width = 0.32\columnwidth]{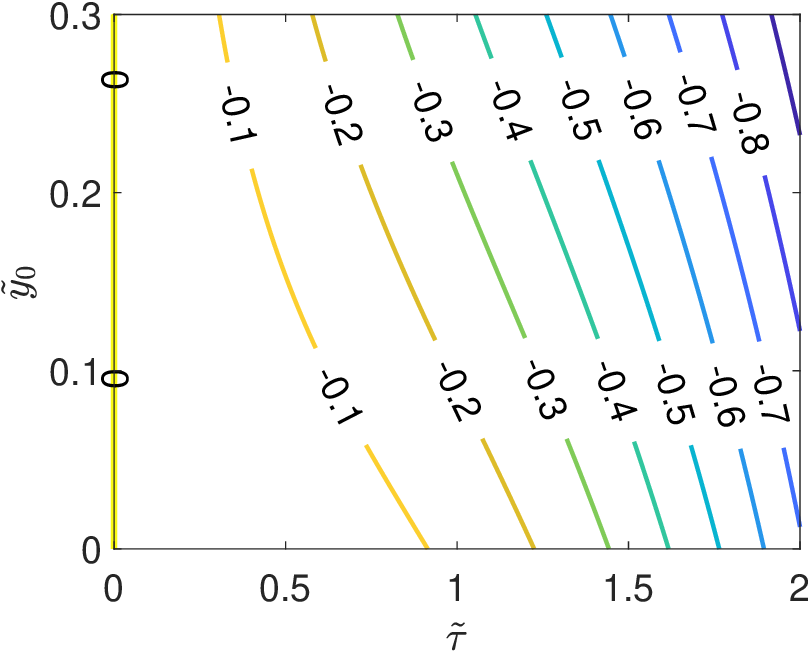}
    \caption{Contour lines of $\Gamma_{v_r}(\cdot,\cdot)$ for different velocity ratios $v_r=\frac{2}{3}$, $v_r=\frac{1}{2}$, $v_r=\frac{2}{5}$ from left to right.}
    \label{fig:solution_new_game2w}
\end{figure}
Here, $v_r$ is chosen as $v_r=\frac{v_p}{v_e}=\frac{1}{1.5}=\frac{2}{3}$, $v_r=\frac{v_p}{v_e}=\frac{1}{2}$ and $v_r=\frac{v_p}{v_e}=\frac{2}{5}$, respectively.
For $v_r\leq \frac{1}{2}$, the $x$-component of the solution \eqref{eq:sol_xP2} does not change its sign in the game set $\mathcal{S}$. For $v_r \geq \frac{1}{2}$ we observe the behavior shown in Figure \ref{fig:solution_new_game2w} on the left and in particular, there exists 
\begin{align}
    (\tilde{\tau}^\#,\tilde{y}_0^\#) \in [0,2\pi) \times [0,\tfrac{\omega_e \rho}{v_e}] \label{eq:Gamma_conditions2}
\end{align}
such that 
\begin{align}
    \Gamma_{v_r}(\tilde{\tau}^\#,\tilde{y}_0^\#)=0  \qquad \text{and} \qquad  \tfrac{\partial}{\partial \tilde{\tau}}\Gamma_{v_r}(\tilde{\tau}^\#,\tilde{y}_0^\#)=0. \label{eq:Gamma_conditions}
\end{align}
The point $(\tilde{\tau}^\#,\tilde{y}_0^\#)$ is highlighted in Figure \ref{fig:solution_new_game2w} on the left.
Moreover, for all $\tilde{y}_0^\times \in (0,\tilde{y}_0^\#)$ there exists a unique $\tilde{\tau}^\times \in (0,\tilde{\tau}^\#)$ with $\Gamma(\tilde{\tau}^\times,\tilde{y}_0^\times)=0$ and where the solution \eqref{eq:sol_xP2} changes its sign in the $x$-component (from negative to positive for $u_e^*=-1$).

Using the coordinate transformations \eqref{eq:scaled_params}, the pair $(\tilde{\tau}^\#,\tilde{y}_0^\#)$ can be equivalently written as $(\tau^\#,y_0^\#)=(\frac{\tilde{\tau}^\#}{\omega_e},\frac{v_e}{\omega_e}\tilde{y}_0^\#) \in (0,\frac{2\pi}{\omega_e}) \times (0,\rho)$. Similarly, $(\tilde{\tau}^\times,\tilde{y}_0^\times)$ can be expressed in terms of the original variables 
$\tau^\times = \frac{1}{\omega_e} \tilde{\tau}^\times \in (0,\tau^\#)$ 
and
$y_0^\times = \frac{v_e}{\omega_e}\tilde{y}_0^\times \in (0,y_0^\#)$, respectively. Based on the discussion in this section we can state the following result.

\begin{theorem} \label{thm:opt_strat_game3}
Consider the game of degree defined in Problem \ref{prob:game_of_degree} with $\frac{v_p}{v_e}> \frac{1}{2}$. Let $\xi_0\in \{\xi \in \mathcal{S}| x= 0 \}$ denote an initial condition in backward time. 
Then, for $u_{e_\tau}^* \in \{-1,1\}$ constant, components of the pursuer's strategy
\eqref{eq:opt_u_p_prob2} and components of the solution \eqref{eq:sol_xP2} that satisfy $u_{e_\tau}^*\in \sign(-x_{\tau}(\tau))$
for $\tau\in \R_{\geq 0}$ characterize components of the solution of Problem \ref{prob:game_of_degree} in backward time.
\hfill $\lrcorner$
\end{theorem}

\begin{proof}
The statement follows from the discussion in this section and from Theorem \ref{thm:opt_strat_game2}. In particular, by construction and as shown in Theorem \ref{thm:opt_strat_game2}, the trajectories \eqref{eq:sol_xP2}, \eqref{eq:adjoint_solution_P2} satisfy \eqref{eq:adjoint_equation} for all $\tau \in \R_{\geq 0}$ and the condition $u_{e_\tau}^*\in \sign(-x_{\tau}(\tau))$ ensures that only the correct components of the trajectory \eqref{eq:sol_xP2} are considered as solutions
\end{proof}

A more detailed representation of the solution in Theorem \ref{thm:opt_strat_game3} in terms of the initial conditions is given in Appendix \ref{ap:characterization_weird_trajectories}.
Here, we avoid lengthy expressions and tedious but straightforward calculations and instead end the section with final remarks.

\begin{remark}
    \label{rem:implicitCondition}
    Unfortunately, an analytic solution of
    $\Gamma_{v_r}(\tilde{\tau},\tilde{y}_0)= 0$, where $\Gamma_{v_r}$ is defined in \eqref{eq:Gamma}, is not available.
    Thus, $\Gamma_{v_r}(\tilde{\tau},\tilde{y}_0)= 0$ should be treated as an implicit condition. 
    However, numerically it can be shown that $\Gamma_{v_r}(\tilde{\tau},\tilde{y}_0)= 0$ only admits solutions $(\tilde{\tau},\tilde{y}_0) \in \R_{>0}^2$ if $v_r > \frac{1}{2}$. Thus, part $\text{(C)}$ of the solution in Figure \ref{fig:solution_game_degree_intro_new} is only present in the case that $v_r = \frac{v_p}{v_e}>\frac{1}{2}$. 
    Similarly, while we cannot solve the equations \eqref{eq:Gamma_conditions} analytically, we can solve it numerically. Then it follows with $y_0^\#=\frac{v_e}{\omega_e}\tilde{y}_0^{\#}$ that
 $\{\xi \in \mathcal{S}|x=0,\ y< y_0^\#\}$ is a universal surface \cite[Sec. 9.5.2]{Lewin2012} and $\{\xi \in \mathcal{S}|x=0,\ y > y_0^\# \}$ is a dispersal surface \cite[Sec. 9.5.2]{Lewin2012}.
    \hfill $\diamond$
\end{remark}

\begin{remark}
    In light of Remark \ref{rem:implicitCondition} and the statement at the end of Section \ref{sec:solutions_ending_on_y_axis}, we see that all singular arcs in the solution to Problem \ref{prob:game_of_degree} lie on the $y$-axis.
    This arrangement of singular arcs is considerably simpler than in the classic surveillance-evasion game of \cite{Lewin1975} with an agile evader but turn-limited pursuer.
    However, the arrangement's dependence on the speed ratio is surprising since such dependence is not present in the solutions to analogous pursuit-evasion and collision-avoidance games with agile pursuers and turn-limited evaders (cf. \cite{Exarchos2014,Exarchos2015,Molloy2020}).
    \hfill $\diamond$
\end{remark}

\begin{remark}
Recall that the point  $y_0^\#=\frac{v_e}{\omega_e}\tilde{y}_0^{\#}$ on the $y$-axis separates the universal surface from the dispersal surface of the game. In the case that $\frac{v_p}{v_e}\leq \frac{1}{2}$, i.e., in the case that \eqref{eq:Gamma_conditions} does not admit a solution satisfying \eqref{eq:Gamma_conditions2}, the switching point is simply given by $y_0^\#=0$. This allows a bearing-only implementation of the evader's optimal strategy
\begin{align}
    u_{e_\tau}^*\in \left\{ \begin{array}{cl}
        \{-1\} & \text{if } \theta \in (0,\pi)  \\
        \{1\}  & \text{if } \theta \in (-\pi,0)  \\
        \{0\} & \text{if } \theta =\pi  \\
        \{-1,1\} & \text{if } \theta =0  
    \end{array} \right.
\end{align}
where $\theta$ is defined as $\theta=\arctan_2(x,y)$ for $\xi \neq 0$.
This characterization is not possible in the case $\frac{v_p}{v_e}> \frac{1}{2}$ where $y_0^\#>0$. Here, the optimal evader strategy additionally depends on distance measurements on the $y-$axis:
\begin{align}
    u_{e_\tau}^*\in \left\{ \begin{array}{cl}
        \{-1\} & \text{if } \theta \in (0,\pi)  \\
        \{1\}  & \text{if } \theta \in (-\pi,0)  \\
        \{0\} & \text{if } \theta =\pi \text{ or } (\theta =0 \text{ and }|\xi|\leq y_0^\#) \\
        \{-1,1\} & \text{if } \theta =0 \text{ and } |\xi|> y_0^\#.
    \end{array} \right.
\end{align}
\hfill $\diamond$
\end{remark}

\section{Illustrative Analysis of Solution of Game of Degree}
\label{sec:analysis_of_solution}

In this section we discuss properties of the solution of the game of degree derived in Section \ref{sec:solution_game_of_degree}, and its connections to recent surveillance-evasion problems \cite{Molloy2023,Weintraub2024,VonMoll2023} that have been posed and solved as optimal-control problems.

\subsection{Specialization to Minimum-Time Escape from a Circular Region for a Dubins Car}

The game of degree discussed in this paper covers the optimal-control problem considered in \cite{Molloy2023, Weintraub2024} of controlling a turn-limited evader so as to escape from a stationary circular region in minimum time as a special case.
Specifically, the solutions derived in \cite{Molloy2023, Weintraub2024} are recovered when the pursuer is stationary, thus with $v_p=0$. 
The solution for this special case of the game is visualized in Figure \ref{fig:stationary}. 
Figure \ref{fig:stationary} shows that the evader's optimal strategy is to turn away from the (stationary) pursuer until either the pursuer is sufficiently far away, or the pursuer is directly behind the evader so that it can switch to traveling in a straight line until the pursuer is sufficiently far away.
\begin{figure}[htb!]
    \centering
    \includegraphics[width = 0.49\columnwidth]{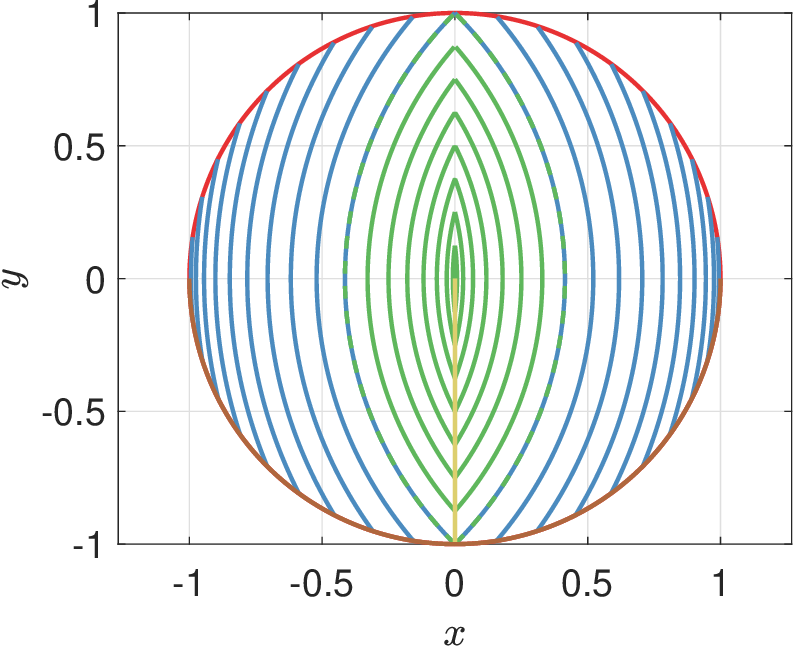}
    \caption{Solution of the game of degree for $\omega_e=1$, $v_p=0$, $\rho=1$ and $v_e=1$.}
    \label{fig:stationary}
\end{figure}
Here, the solution is not only symmetric with respect to the $y$-axis, but also shows symmetry properties with respect to the $x$-axis. Interestingly, the solution in \cite{Molloy2023}, \cite{Weintraub2024} is derived (via optimal control or geometric arguments) using a coordinate system centered on the (stationary) pursuer, while this paper provides a solution (via differential-game techniques) using a coordinate system centered on the evader.

\subsection{Game Solutions in Inertial Coordinate Frame}

Instead of the $\xi$-coordinates used to derive a solution of the game of degree, we can visualize solutions in the original coordinates $\xi_p,\xi_e \in \R^2$. Figure \ref{fig:orig_coordinates} shows the solution of the game of game of degree for three different initial conditions for the parameter selection $v_e=1.5$, $v_p=1$, $w_e=1$ and $\rho=1$.
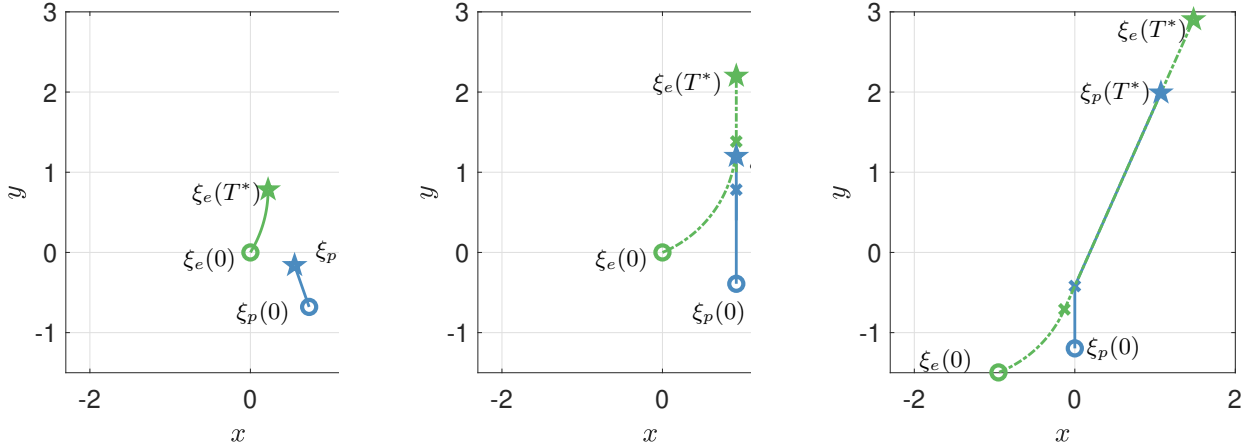
\begin{figure}[htb!]
    \centering
    \begin{overpic}[width = 0.33\columnwidth]{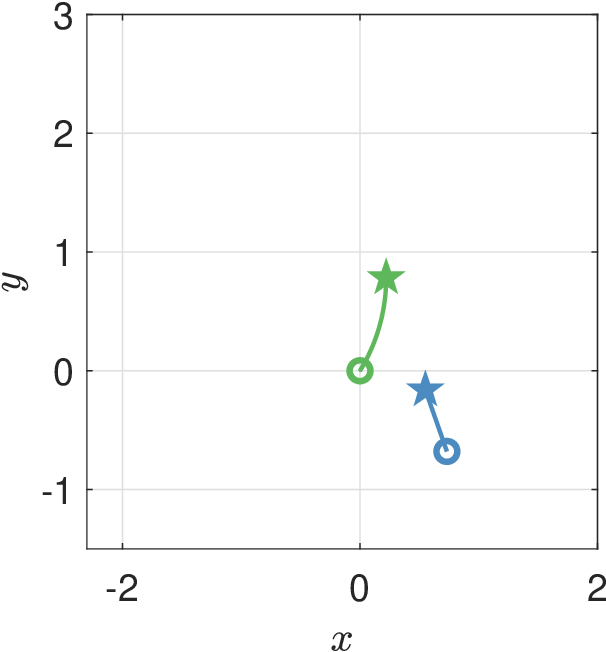}
    \put(40,40){\small{$\xi_e(0)$}}
    \put(52,28){\small{$\xi_p(0)$}}
    \put(42,55){\small{$\xi_e(T^*)$}}
    \put(70,42){\small{$\xi_p(T^*)$}}
    \end{overpic}
    \begin{overpic}[width = 0.33\columnwidth]{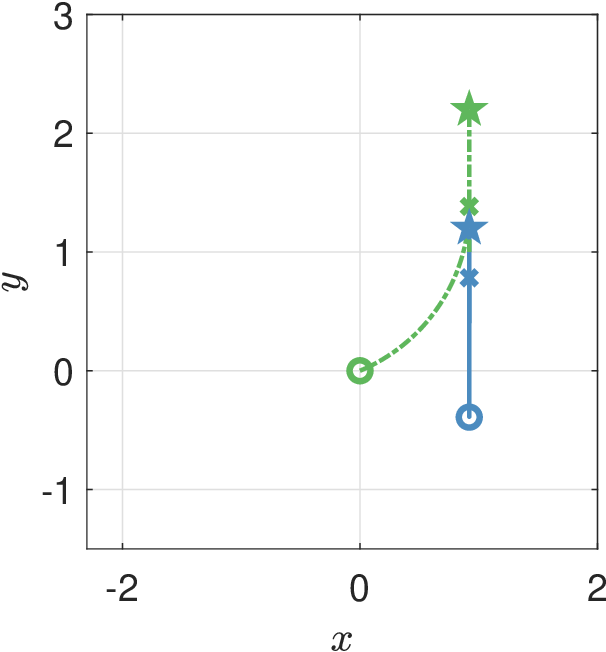}
    \put(40,40){\small{$\xi_e(0)$}}
    \put(62,28){\small{$\xi_p(0)$}}
    \put(53,80){\small{$\xi_e(T^*)$}}
    \put(75,62){\small{$\xi_p(T^*)$}}
    \end{overpic}
    \begin{overpic}[width = 0.33\columnwidth]{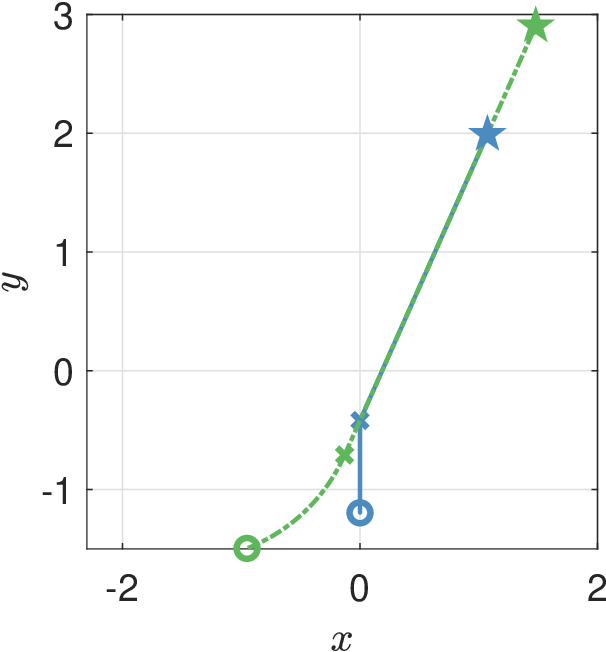}
    \put(20,17){\small{$\xi_e(0)$}}
    \put(58,20){\small{$\xi_p(0)$}}
    \put(65,92){\small{$\xi_e(T^*)$}}
    \put(56.5,78){\small{$\xi_p(T^*)$}}
    \end{overpic}
    \caption{Solutions of the game of degree in the inertial coordinate frame for different initial conditions.}
    \label{fig:orig_coordinates}
\end{figure}
Here, we show the solutions corresponding to the cases (A), (B) and (C) from left to right. Points where the strategy of the pursuer and the evader switch are indicated through $\times$ along the solution. This is the case when $\xi(t)=\xi_p(t)-\xi_e(t)$ reaches the $y$-axis. In Figure \ref{fig:orig_coordinates} on the right, the pursuer positions itself in front of the evader before the evader overtakes the pursuer to end the game.

\subsection{Optimal Value Function Compared to Pure Pursuit}

Finally, we examine how the optimal (Nash-equilibrium) strategies for the evader and pursuer derived here compare with a setting where the evader employs its optimal strategy against a pursuer that uses pure pursuit (i.e., the pursuer selects its heading angle by simply pointing at the evader).
This setting is similar to that considered in \cite{VonMoll2023} where the evader seeks to minimize its time under surveillance \emph{knowing} that the pursuer employs pure pursuit.
However, in contrast to \cite{VonMoll2023}, the evader we consider is playing its Nash-equilibrium strategy that guards against a maximizing pursuer.
To visualize the differences between when the pursuer uses its optimal (Nash-equilibrium) strategy and when it uses pure pursuit, we visualize the game value (stopping time) of the two strategies for two speed ratios in Figures \ref{fig:solution_pure_pursuit_1} and \ref{fig:solution_pure_pursuit_2}; that is, we visualize the time it takes until the evader is out of the surveillance region.
\begin{figure}[htb!]
    \centering
    \includegraphics[width = 0.49\columnwidth]{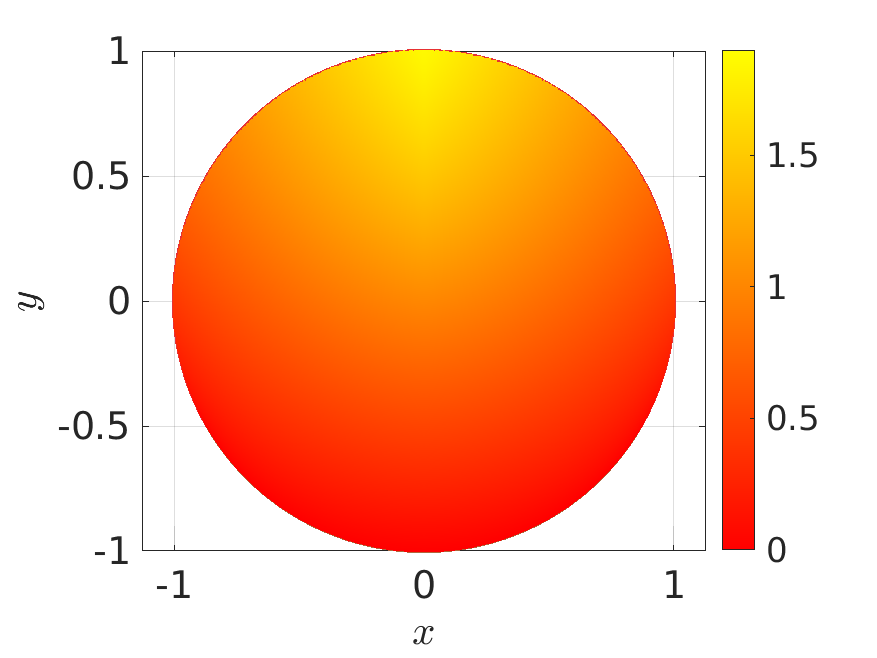}
    \includegraphics[width = 0.49\columnwidth]{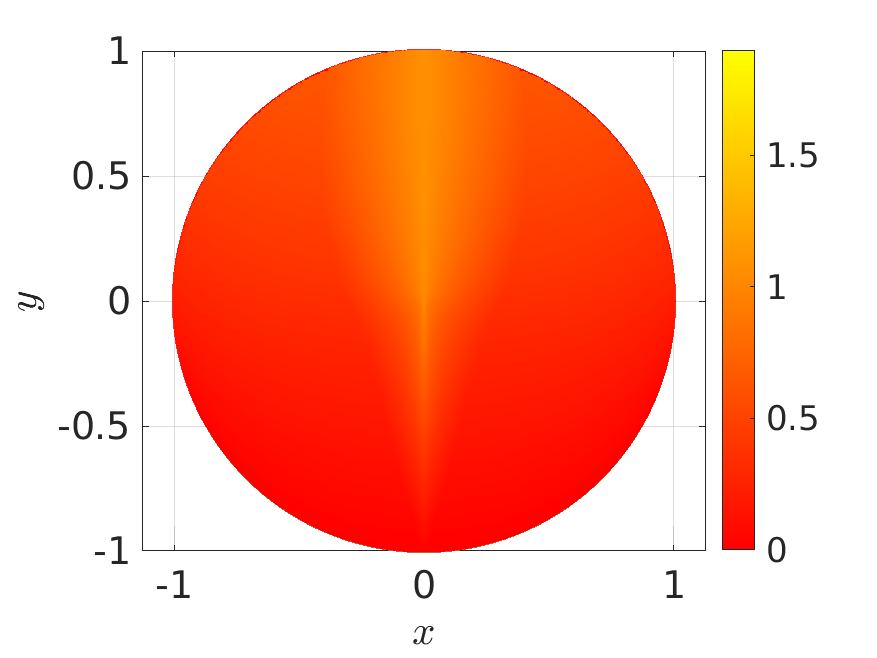}
    \caption{Visualization of the value function under optimal play \eqref{eq:problem_constraints} (left) and pure pursuit (right) using the parameter selection $\omega_e=1$, $v_p=1$, $\rho=1$, $v_e=2$. Here, the maximal values are given by $1.85$ (left) and $1.08$ (right), respectively.}
    \label{fig:solution_pure_pursuit_1}
\end{figure}
\begin{figure}[htb!]
    \centering
    \includegraphics[width = 0.49\columnwidth]{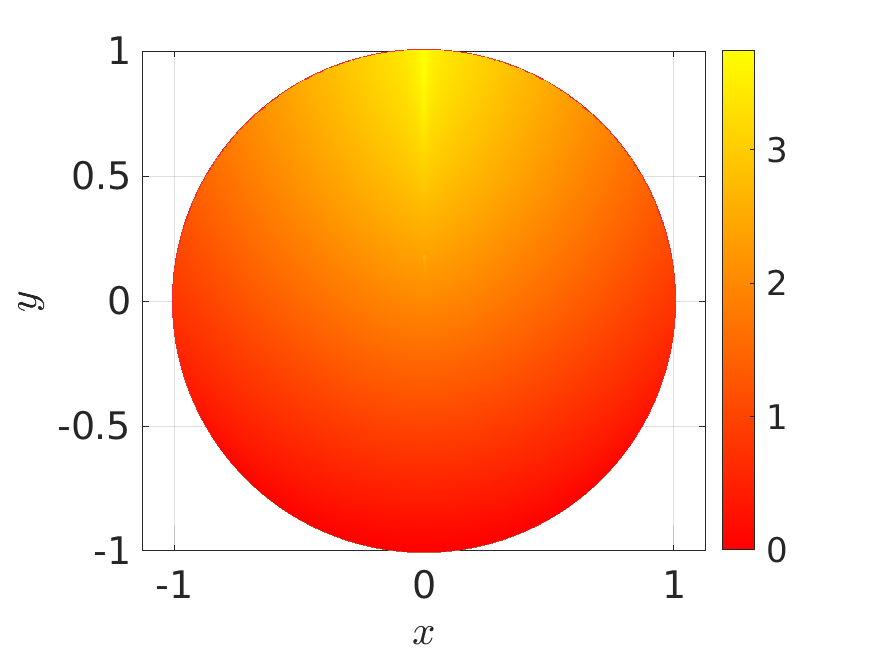}
    \includegraphics[width = 0.49\columnwidth]{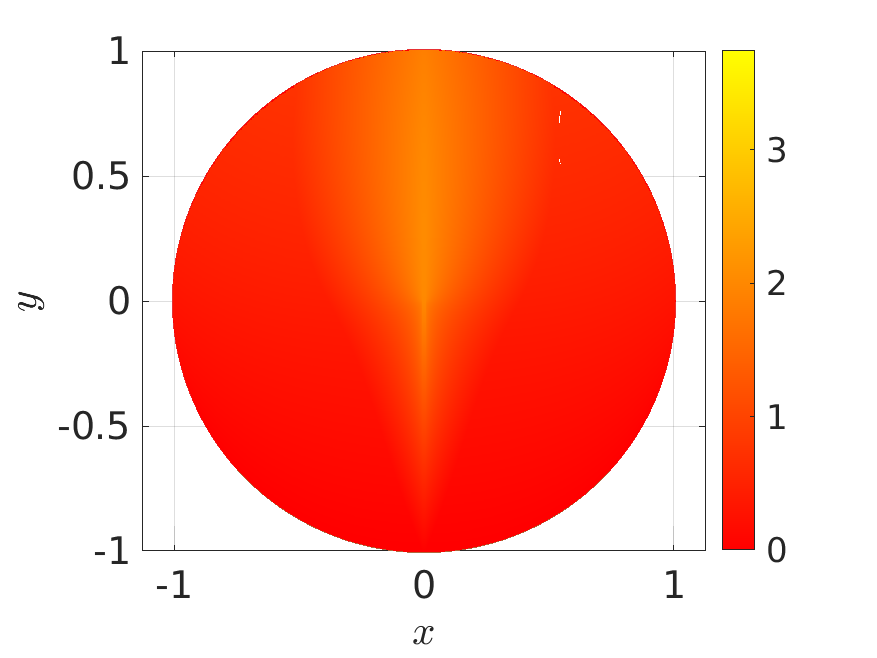}
    \caption{Visualization of the value function under optimal play \eqref{eq:problem_constraints} (left) and pure pursuit (right) using the parameter selection $\omega_e=1$, $v_p=1$, $\rho=1$, $v_e=1.5$. Here, the maximal values are given by $3.73$ (left) and $2.04$ (right), respectively.}
    \label{fig:solution_pure_pursuit_2}
\end{figure}

From Figures \ref{fig:solution_pure_pursuit_1} and \ref{fig:solution_pure_pursuit_2}, we see that when both the evader and pursuer employ their optimal (Nash-equilibrium) strategies, the time it takes for the evader to leave the surveillance region is almost always greater than that if the pursuer uses pure pursuit.
In particular, there is a significant difference in times when the game starts with the pursuer in front of the evader (i.e., when $\xi_0 \in \{\xi \in \mathcal{S}| x\geq 0 \}$) because pure pursuit leads to the pursuer inadvertently assisting the evader in leaving the surveillance region by not anticipating that the evader's strategy is to place the pursuer behind it which is easiest if the pursuer is closer.

\section{Conclusion}
\label{sec:conclusion}

We posed and solved a novel surveillance-evasion differential game, in both its game-of-kind and game-of-degree forms, between an agile pursuer and a turn-limited evader.
The solution of the game of degree is surprisingly complex in the case of a faster evader and slower pursuer, with the characterization of the solution depending on the speed ratio. 
Future work will focus on the extension of the game of degree and its solution in the case of multiple pursuers. 
Additionally, obstacles leading to obstructions in the game set will be considered.

\section*{Funding Sources}

The authors are supported by the ANU CIICADA lab and the United States Air Force Office of Scientific Research under Grant No. FA2386-24-1-4014. 

\begin{appendix}
    
\section*{Appendix}

\section{Derivation of Dynamics \eqref{eq:dynamics}} \label{sec:coordinate_transformation}

To derive the dynamics \eqref{eq:dynamics} we take the derivative of the expression \eqref{eq:coordinate_transformation_xi} leading to 
\begin{align}
    \dot{\xi} =
 \dot{\theta}_e    \left[\begin{array}{rr}
     -\sin(\theta_e) & -\cos(\theta_e)  \\
      \cos(\theta_e) & -\sin(\theta_e) 
 \end{array} \right] \left[ 
\begin{array}{c}
     x_p -x_e \\
     y_p -y_e
\end{array}
 \right]  + 
    \left[\begin{array}{rr}
     \cos(\theta_e) & -\sin(\theta_e)  \\
     \sin(\theta_e) & \cos(\theta_e) 
 \end{array} \right] \left[ 
\begin{array}{c}
     \dot{x}_p -\dot{x}_e \\
     \dot{y}_p -\dot{y}_e
\end{array}
 \right].  \label{eq:xi_derivation1}
\end{align}

With 
\begin{align}
\left[ 
\begin{array}{c}
     x_p -x_e \\
     y_p -y_e
\end{array} \right] =
     \left[\begin{array}{rr}
     \cos(\theta_e) & \sin(\theta_e)  \\
     -\sin(\theta_e) & \cos(\theta_e) 
 \end{array} \right] \left[ 
\begin{array}{c}
     x \\
     y
\end{array} \right],
\end{align}
which follows from \eqref{eq:coordinate_transformation_xi}, and with the dynamics \eqref{eq:cartesianDynamicsPursuer} and \eqref{eq:cartesianDynamics},
expression \eqref{eq:xi_derivation1} can be rewritten as
\begin{align*}
    \dot{\xi} &=
 \omega_e u_e    \left[\begin{array}{rr}
     -\sin(\theta_e) & -\cos(\theta_e)  \\
     \cos(\theta_e) & -\sin(\theta_e) 
 \end{array} \right]
 \left[\begin{array}{rr}
     \cos(\theta_e) & \sin(\theta_e)  \\
     -\sin(\theta_e) & \cos(\theta_e) 
 \end{array} \right]
 \left[ 
\begin{array}{c}
     x \\
     y
\end{array}
 \right]  + 
    \left[\begin{array}{rr}
     \cos(\theta_e) & -\sin(\theta_e)  \\
     \sin(\theta_e) & \cos(\theta_e) 
 \end{array} \right] \left[ 
\begin{array}{c}
     v_p \cos \theta_p -v_e \cos \theta_e  \\
     v_p \sin \theta_p - v_e \sin \theta_e
\end{array}
 \right] \\
  &=
 \omega_e u_e    \left[\begin{array}{rr}
     0 & -1  \\
     1 & 0 
 \end{array} \right]
 \left[ 
\begin{array}{c}
     x \\
     y
\end{array}
 \right] +
 \left[ 
\begin{array}{c}
v_p \cos \theta_e \sin \theta_p - v_p \cos \theta_p \sin \theta_e \\
- v_e \cos^2 \theta_e + v_p \cos \theta_p \cos \theta_e - v_e \sin^2 \theta_e  + v_p \sin \theta_p \sin \theta_e 
\end{array}
 \right] \\
&=     \left[ 
\begin{array}{c}
    -\omega_e u_e y \\
     \omega_e u_e x
\end{array}
 \right] +
 \left[ 
\begin{array}{c}
     v_p \sin(\theta_p-\theta_e)   \\
     v_p \sin(\theta_p-\theta_e) -v_e
\end{array}
 \right].
\end{align*}
Thus, with $u_p=\theta_p-\theta_e$ the dynamics \eqref{eq:f_dynamics} are recovered.

\section{Auxiliary Results} \label{ap:auxilliary_results}

In this section, we collect proofs of statements given in the paper. Since the proofs are not insightful, they are reported here for completeness but not in the main part of the paper. 

\begin{proofof} \emph{Lemma \ref{lem:strategy_on_y_axis}}:
As a first step, consider the vector
\begin{align}
    {\bf n}(\xi) = \frac{\xi}{|\xi|}, \qquad \xi \neq 0, 
\end{align}
i.e., a vector of length one pointing to the target set $\mathcal{C}$. Hence,  $f(\xi,u_e,u_p)^\top {\bf n}(\xi)$
indicates how fast $\xi(t)$ approaches (or goes away from) the target set. 
Now,
observe that for $u_e(t)=0$ the dynamics \eqref{eq:f_dynamics} reduce to 
\begin{align}
    f(\xi, 0, u_p)
    &=
    \begin{bmatrix}
        v_p \sin(u_p) \\
        - v_e + v_p \cos(u_p) 
    \end{bmatrix}.
\end{align}
From this expression we note that for $\xi_0\in \{\xi \in \mathcal{S}| x=0, \ y< 0 \}$ the best response of the pursuer to minimize 
\begin{align}
f(\xi_0,0,u_p)^\top {\bf n}(\xi_0) = (-v_e+v_p\cos(u_p)) \frac{y}{|y|} = - (-v_e+v_p)
\end{align}
is based on the strategy $u_p(t)=0$.
This leads to $f(\xi(t),0,0)^\top \left[\begin{smallmatrix}
    0 \\ -1
\end{smallmatrix}\right]=v_e-v_p$ and the solution satisfies \eqref{eq:solution_y_axis}.
In this case, the game terminates at time $T=(-\rho-y_0) (v_p-v_e) \in \R_{\geq 0}$ with $\xi(T)=[0,-\rho]^\top$.

Since the optimal response to $u_e=0$ is $u_p=0$, to have a solution $\xi(\cdot)$ leaving the $y$-axis, there needs to exist times $t_1,t_2\in \R_{\geq 0}$, $0 \leq t_1<t_2$ such  that $u_e(t)>0$ or $u_e(t)<0$ for all $t\in [t_1,t_2]$.
For the sake of a contradiction, assume that $u_e(t) >0$ for $t \in [t_1,t_2]$ and without loss of generality $y(t)<0$ for all $t\in[t_1,t_2]$ (since $y(0)<0$ and $y(t)$ is monotonically decreasing for $u_p=0$ and $u_e=0$). Then, with the pursuer's (potentially sub-optimal) strategy defined through the condition $x(t)=|\xi(t)| \sin(u_p(t)+\pi)$, $y(t)=|\xi(t)|\cos(u_p(t)+\pi)$ 
for all $t\in [t_1,t_2]$, 
it holds that 
\begin{align*}
    f(\xi(t), u_e(t), u_p(t))^\top {\bf n}(\xi(t))
    &=
    \begin{bmatrix}
        -\omega_e y u_e(t) + v_p \sin (u_p(t))\\
        \omega_e x u_e(t) - v_e + v_p \cos (u_p(t))
    \end{bmatrix}^\top {\bf n}(\xi(t)) \\
    &= v_p \sin (u_p(t)) \frac{x}{|\xi(t)|} + (-v_e + v_p \cos (u_p(t))) \frac{y(t)}{|\xi(t)|} \\
    &= v_p  \sin (u_p(t)) \sin (u_p(t)+\pi)  -v_e \frac{y(t)}{|\xi(t)|} + v_p \cos (u_p(t))\cos (u_p(t)+\pi) \\
    &= -v_p    -v_e \frac{y(t)}{|\xi(t)|} < -v_p    +v_e 
\end{align*}
for all $t\in[t_1,t_2]$ such that $x(t)\neq 0$. Hence, the strategy $u_e(t)>0$ is not optimal for the evader. A similar contradiction (using symmetry arguments as in Corollary \ref{cor:symmetry}) can be used to show that $u_e(t) <0$ for $t\in [t_1,t_2]$ is not optimal for the evader, which completes the proof.
\end{proofof}

\section{Characterization of Trajectories in Theorem \ref{thm:opt_strat_game3}} \label{ap:characterization_weird_trajectories}

In this part we give a detailed representation of the relevant parts of the functions visualized in Figure \ref{fig:solution_new_game2w}. For the derivation we assume that $\frac{v_e}{v_p}> \frac{1}{2}$ and we recall the definition of
$(\tilde{\tau}^\#,\tilde{y}_0^\#)$
through \eqref{eq:Gamma_conditions2} and \eqref{eq:Gamma_conditions}.
We focus on $\tilde{y}_0^\times \in (0,\tilde{y}_0^\#)$ and $\tilde{\tau}^\times \in (0,\tilde{\tau}^\#)$ with $\Gamma(\tilde{\tau}^\times,\tilde{y}_0^\times)=0$. With these definitions, for the case $u_{e_\tau}^*=-1$, the $x$-component of the solution solution \eqref{eq:sol_xP2} initialized through $y_0^\times = \frac{v_e}{\omega_e}\tilde{y}_0^\times$ changes its sign from negative to positive at time 
$\tau^\times = \frac{1}{\omega_e} \tilde{\tau}^\times.$ 
To extract the component of the trajectory \eqref{eq:sol_xP2} that is relevant for the game of degree we consider the coordinate transformation $\sigma=\tau-\tau^\times$
and we define
\begin{align}
    y_0^\times = y_0 \cos(\tau^\times \omega_e) + \frac{v_e}{\omega_e}  \sin(\tau^\times \omega_e ) - \tau^\times \omega_e v_p \cos(\tau^\times \omega_e ) \label{eq:y_0_equation}
\end{align}
denoting the $y$-component of the solution \eqref{eq:sol_xP2} at time $\tau^\times$ initialized on the $y$-axis $\xi_0=\left[\begin{smallmatrix}
    0\\
    y_0
\end{smallmatrix}\right]$, $y_0\in (0,\frac{v_e}{w_e} \tilde{y}^\#)$.

By construction, the $x$-component of the solution trajectory \eqref{eq:sol_xP2} is equal to zero at time $\tau^\times$ and we can thus define the new initial condition
\begin{align}
\xi_0^\times = \left[\begin{array}{c}
     0  \\
     y_0^\times 
\end{array} \right]
\end{align}
Solving \eqref{eq:y_0_equation} for $y_0$ provides the expression
\begin{align}
    y_0  = \frac{y_0^\times   - \frac{v_e}{\omega_e} \sin(\tau^\times \omega_e) + \tau^\times v_p \cos(\tau^\times \omega_e)}{\cos(\tau^\times \omega_e)}
\end{align}
which allows us to rewrite the solution \eqref{eq:sol_xP2} in terms of the time argument $\sigma$ and the initial condition $\xi_0^\times$:
\begin{align}
    &\tilde{\xi}_{\tau}(\sigma; \xi_0^\times,u_{e_\tau}^*,u_{p_\tau}^*(\cdot+\tau^\times)) = \xi(\sigma+\tau^\times;\xi_0,u_{e_\tau}^*,u_{p_\tau}^*(\cdot+\tau^\times)) \label{eq:sol_p3} \\
    &=
    \begin{bmatrix}
        u_{e_\tau}^* (y_0 \sin((\sigma+\tau^\times) \omega_e) + \frac{v_e}{\omega_e} - (\sigma+\tau^\times) v_p \sin((\sigma+\tau^\times) \omega_e) -\frac{v_e}{\omega_e}\cos((\sigma+\tau^\times) \omega_e)) \\
        y_0 \cos((\sigma+\tau^\times) \omega_e) + \frac{v_e}{\omega_e} \sin((\sigma+\tau^\times) \omega_e) - (\sigma+\tau^\times) v_p \cos((\sigma+\tau^\times) \omega_e)
    \end{bmatrix} \nonumber \\
    &=
    \begin{bmatrix}
        u_{e_\tau}^* (\frac{y_0^\times   - \frac{v_e}{\omega_e} \sin(\tau^\times \omega_e) + \tau^\times v_p \cos(\tau^\times \omega_e)}{\cos(\tau^\times \omega_e)} \sin((\sigma+\tau^\times) \omega_e) + \frac{v_e}{\omega_e} - (\sigma+ \tau^\times) v_p \sin((\sigma+\tau^\times) \omega_e) -\frac{v_e}{\omega_e}\cos((\sigma+\tau^\times) \omega_e)) \\
        \frac{y_0^\times   - \frac{v_e}{\omega_e} \sin(\tau^\times \omega_e) + \tau^\times v_p \cos(\tau^\times \omega_e)}{\cos(\tau^\times \omega_e)} \cos((\sigma+\tau^\times) \omega_e) + \frac{v_e}{\omega_e} \sin((\sigma+\tau^\times) \omega_e) - (\sigma+\tau^\times) v_p \cos((\sigma+\tau^\times) \omega_e)
    \end{bmatrix}. \nonumber
\end{align}
Similarly, the adjoint variables satisfy 
\begin{align*}
\tilde{\bp}_\tau(\sigma) =    \bp_\tau(\sigma+\tau^\times) 
&= \left[ \begin{array}{rr}
        \cos(-\omega_e u_e^* (\sigma+\tau^\times)) & -\sin(-\omega_e u_e^* (\sigma+\tau^\times)) \\
        \sin(-\omega_e u_e^* (\sigma+\tau^\times))  & \cos(-\omega_e u_e^*(\sigma+\tau^\times)) 
    \end{array}\right]
    \left[\begin{array}{c}
         0  \\
         \tfrac{1}{v_e-v_p}  
    \end{array} \right] \\
    &= \left[ \begin{array}{rr}
        \cos(-\omega_e u_{e_\tau}^* \sigma) & -\sin(-\omega_e u_{e_\tau}^* \sigma) \\
        \sin(-\omega_e u_{e_\tau}^* \sigma)  & \cos(-\omega_e u_{e_\tau}^* \sigma) 
    \end{array}\right] 
    \left[ \begin{array}{rr}
        \cos(-\omega_e u_{e_\tau}^* \tau^\times) & -\sin(-\omega_e u_{e_\tau}^* \tau^\times) \\
        \sin(-\omega_e u_{e_\tau}^* \tau^\times)  & \cos(-\omega_e u_{e_\tau}^* \tau^\times) 
    \end{array}\right]
    \left[\begin{array}{c}
         0  \\
         \tfrac{1}{v_e-v_p}  
    \end{array} \right]
\end{align*}
and we can define the new initial condition as
\begin{align}
        \tilde{\bp}_0=
        \tfrac{1}{v_e-v_p}
        \left[ \begin{array}{rr}
        -\sin(-\omega_e u_{e_\tau}^* \tau^\times) \\
         \cos(-\omega_e u_{e_\tau}^* \tau^\times) 
    \end{array}\right].
\end{align}

\end{appendix}

\bibliographystyle{abbrv}


\end{document}